\documentclass[a4paper,twoside,11pt]{article}

\usepackage{BeaCed_article}

\title{
\mbox{From Cycle Rooted Spanning Forests to the Critical Ising Model:}\\
an Explicit Construction}
\author{B\'eatrice de Tili\`ere
\thanks{{\small
Laboratoire de Probabilit\'es et Mod\`eles Al\'eatoires, Universit\'e Pierre et Marie Curie, 4 place Jussieu, 
F-75005 Paris.}
{\small\texttt{beatrice.de\_tiliere@upmc.fr.}}
{\small Institut de Math\'ematiques, Universit\'e de Neuch\^atel, Rue Emile-Argand 11, 
CH-2007 Neuch\^atel. Supported in part by the Swiss National Foundations grant 200020-120218.}
}}
\date{}

\begin{document}

\maketitle

\begin{abstract}
Fisher established an explicit correspondence between the $2$-dimensional Ising
model defined on a graph $G$ and the dimer model defined on a decorated
version~$\GD$ of this graph \cite{Fisher}. In this paper we explicitly relate
the dimer model associated to the critical Ising model and critical
cycle rooted spanning forests (CRSFs). This relation is established through
characteristic polynomials, whose definition only depends on the respective
fundamental domains, and which encode the combinatorics of the model. We
first
show a matrix-tree type theorem establishing that the dimer characteristic
polynomial counts CRSFs of the decorated fundamental domain $\GD_1$. Our main
result consists in explicitly constructing CRSFs of~$\GD_1$ counted by the dimer
characteristic polynomial, from CRSFs of $G_1$ where edges are assigned
Kenyon's critical weight function \cite{Kenyon3}; thus proving a relation on
the level of configurations between two well known 2-dimensional critical
models.
\end{abstract}

\section{Introduction}

In \cite{Fisher}, Fisher established an explicit correspondence between the
$2$-dimensional Ising model defined on 
a graph $G$ and the dimer model defined on a decorated version~$\GD$ of this
graph, known as the \emph{Fisher graph} of $G$. Since then, dimer techniques
have been a powerful
tool for tackling the Ising model, see for example the book of \cite{McCoyWu}. More recently, in 
\cite{BoutillierdeTiliere:iso_perio,BoutillierdeTiliere:iso_gen},
we prove fundamental results for the dimer model corresponding
to a large class of \emph{critical} Ising models, by proving
explicit formulae for the free energy and for the Gibbs measure.

Critical Ising models we consider are defined on graphs satisfying a geometric
property called \emph{isoradiality}. When
the underlying isoradial graph $G$ is infinite and $\ZZ^2$-periodic, then so is
the Fisher graph $\GD$, and we let $\GD_1=\GD/\ZZ^2$ be the \emph{fundamental
domain}. The key object involved in the
explicit expressions of \cite{BoutillierdeTiliere:iso_perio} for the free energy
and the Gibbs measure is the \emph{critical dimer characteristic 
polynomial}, whose definition only depends on the fundamental domain $\GD_1$.
This polynomial
is a generating function 
for configurations related to super-imposed dimer configurations of $\GD_1$,
referred to as `double-dimer' configurations. 
By Fisher's correspondence, this implies that the dimer characteristic
polynomial is a generating function for `double-Ising' 
configurations.

In \cite{BoutillierdeTiliere:iso_perio}, we prove that the dimer characteristic
polynomial is equal, up to a constant, to the
\emph{critical Laplacian characteristic polynomial} of $G_1=G/\ZZ^2$, where
edges of $G_1$ are assigned Kenyon's critical weight function \cite{Kenyon3}.
Using a generalization of Kirchhoff's matrix tree theorem due to Forman
\cite{Forman}, the
latter is shown to be a generating function for \emph{cycle rooted
spanning forests} (CRSFs) of $G_1$. This suggests the existence of an explicit
relation between 
`double-Ising' configurations
and CRSFs, which we were not able to find in
\cite{BoutillierdeTiliere:iso_perio}. The first result of this paper is a
matrix-tree type theorem, proving that the critical dimer characteristic
polynomial is a generating function for CRSFs of the Fisher graph $\GD_1$, see
Theorem \ref{thm:matrix-tree} of Section \ref{sec:matrix-tree}.
Then, the main result of this paper can loosely be stated as follows, refer to
Theorem~\ref{thm:main2} of Section \ref{sec:statementmain} and Theorem
\ref{thm:correspondence} of Section \ref{sec:explicitcorrespondence} for more
precise statements.
\begin{thm}\label{thm:main}
Consider a critical Ising model defined on an infinite,
$\ZZ^2$-periodic isoradial graph $G$. Then,
there exists an explicit way of constructing CRSFs of $\GD_1$ counted by the
critical dimer characteristic polynomial, from CRSFs of $G_1$ counted by the
critical Laplacian characteristic
polynomial.
\end{thm}
This exhibits an explicit relation, on the level of configurations, 
between two well known models of
statistical mechanics: the Ising model and CRSFs at criticality. Note that such
a relation was already
suspected by Messikh \cite{messikh}. Before giving an outline of the paper, let
us make a few comments.
\begin{itemize}
\item The main contribution of this paper is the proof of Theorem
\ref{thm:main}, where we actually provide the explicit construction. 
\item The partition functions (weighted sum of configurations) of both
models can be expressed from their
respective characteristic polynomials, so that it is actually stronger to work with characteristic polynomials.
\item Working with graphs embedded on the torus has the advantage of avoiding boundary issues, but has the
additional difficulty of involving the geometry of the torus, with non-trivial  
cycles occurring in configurations. 
Working on finite pieces of infinite graphs, and precisely specifying boundary
conditions, would certainly explicitly relate double-dimer configurations
and spanning trees. 
\item Spanning trees are a well suited object for defining a height function and
prove Gaussian fluctuations. Thus, it might be that
Theorem \ref{thm:main} could be 
used to prove results which are numerically described in the paper
\cite{Wilson}.
\end{itemize}

\textbf{Outline of the paper}
\begin{enumerate}
\item[Section \ref{sec:intro1}:] Definition of the critical Ising model and of
the dimer model. Description of Fisher's correspondence relating the two.
\item[Section \ref{sec:characteristicpolynomials}:] Definition of the
critical dimer and Laplacian characteristic
polynomials. Relation between the Laplacian characteristic polynomial
and CRSFs. 
\item[Section \ref{sec:essential}:] Statement and proof of Theorem
\ref{thm:matrix-tree} establishing that the critical dimer characteristic
polynomial is a generating function for CRSFs of the Fisher graph $\GD_1$.
Precise statement of Theorem \ref{thm:main}.
\item[Section \ref{sec:general}:] Definition and properties of \emph{licit
primal/dual edge moves}, which are one of the key ingredients of the
correspondence. 
\item[Section \ref{sec:CORRES}:] Explicit
construction of CRSFs of $\GD_1$ from CRSFs of $G_1$ and proof of Theorem
\ref{thm:main}.
\end{enumerate}

\section{Critical Ising model and dimer model}\label{sec:intro1}

In this section, we define the $2$-dimensional critical Ising model, the dimer
model and describe Fisher's correspondence relating the two.

\subsection{Critical Ising model}

Let $G=(V(G),E(G))$ be a finite, planar, unoriented graph, together with a
collection of positive real numbers $J=(J_e)_{e\in E(G)}$ indexed by the edges
of $G$. The {\em Ising model on $G$ with coupling constants $J$} is defined as
follows. A \emph{spin configuration} $\sigma$ of $G$ is a function of the
vertices of $G$ with values in $\{-1,+1\}$. The probability of occurrence of a
spin configuration $\sigma$ is given by the {\em Ising Boltzmann measure},
denoted $\PPising$:
\begin{equation*}
\PPising(\sigma)=\frac{1}{\Zising}\exp\left(\sum_{e=uv\in
E(G)}J_e\sigma_u\sigma_v\right),
\end{equation*}
where 
$
\Zising=\sum_{\sigma\in\{-1,1\}^{V(G)}}\exp\left(\sum_{e=uv\in
E(G)}J_e\sigma_u\sigma_v\right),
$
is the {\em Ising partition function}. 

We consider Ising models defined on a class of embedded graphs which have an additional 
property called {\em isoradiality}. A graph $G$ is said to be {\em isoradial}~\cite{Kenyon3}, 
if it has an embedding in the plane such that every face is inscribed in a circle of radius~1. 
We ask moreover that
all circumcenters of the faces are in the closure of the faces. From now on, when we speak 
of the graph $G$, we mean the graph together with a particular isoradial embedding in the
plane. Examples of isoradial graphs are the square and the honeycomb lattice.
Refer to Figure \ref{fig:isingcrit} (left) for a more general example of
isoradial graph.

To such a graph is naturally associated the {\em diamond graph}, denoted by
$\GR$:
vertices of $\GR$ consist in the vertices
of $G$ and the circumcenters of the faces of $G$ (which are also the dual
vertices of $G^*$); the circumcenter of each
face is then joined to all vertices which are on the boundary of this face, see
Figure \ref{fig:isingcrit} (center). Since $G$ is isoradial,
all faces of $\GR$ are side-length-$1$ rhombi. Moreover, each edge $e$ of
$G$ is the diagonal of exactly one rhombus of $\GR$; we let $\theta_e$ be the
half-angle of the rhombus at the vertex it has in common with $e$, see Figure
\ref{fig:isingcrit} (right). 

\begin{figure}[ht]
\begin{center}
\includegraphics[width=\linewidth]{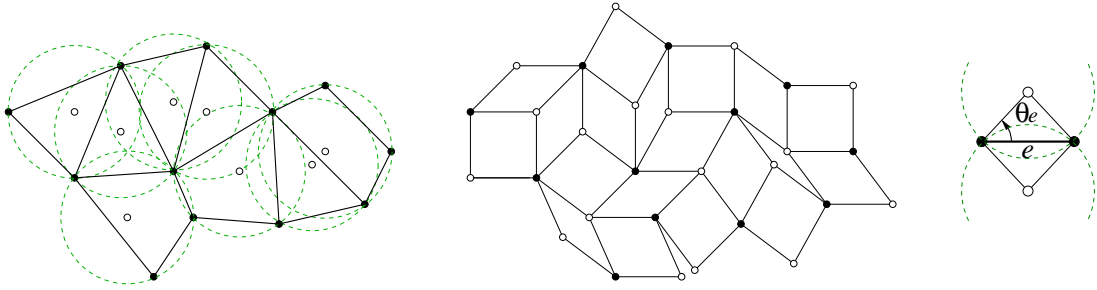}
\caption{Left: example of isoradial graph. Center: corresponding diamond graph.
Right: rhombus half-angle associated to 
an edge $e$ of the graph.}
\label{fig:isingcrit}
\end{center}
\end{figure}

The same construction can be done for infinite and toroidal isoradial graphs, in
which case the embedding is in the plane or in the torus.

It is then natural to choose the coupling constants $J$ of the Ising model to
depend on the geometry of the embedded graph: let us
assume that $J_e$ is a function of $\theta_e$, the rhombus half-angle assigned
to the edge $e$. 

We impose two more conditions on the coupling constants. First, we ask that the
Ising model on 
$G$ with coupling constants $J$ as above is {\em $Z$-invariant}, that is,
invariant under {\em star-triangle} transformations of the underlying graph.
Next, we impose that the Ising model satisfies a generalized form of {\em
self-duality}. These conditions completely determine the coupling constants $J$,
known as 
{\em critical coupling constants}: for every edge $e$ of $G$,
\begin{equation*}
  J(\theta_e) = \frac{1}{2} \log \left(\frac{1+\sin \theta_e}{\cos \theta_e}\right).
\end{equation*}
The $Z$-invariant Ising model on an isoradial graph with this particular choice
of coupling 
constants is referred to as  \emph{critical Ising model}. This
model was introduced by Baxter in \cite{Baxter:Zinv}. A more detailed definition
is given in \cite{BoutillierdeTiliere:iso_perio}.

\subsection{Dimer model}

Let $\GD=(V(\GD),E(\GD))$ be a finite, planar, unoriented graph, and suppose
that edges of $\GD$ are assigned a 
positive weight function $\nu=(\nu_e)_{e\in E(\GD)}$. The {\em dimer model on
$\GD$ with weight function $\nu$} is defined as follows.

A {\em dimer configuration} $M$ of $\GD$, also called {\em perfect matching}, is
a subset of edges 
of $\GD$ such that every vertex is incident to exactly one edge of $M$. Let
$\M(\GD)$ be the set of dimer configurations of the graph $\GD$. The probability
of occurrence of a dimer configuration $M$ is given by the {\em dimer Boltzmann
measure}, denoted $\PPdimer$:
\begin{equation*}
 \PPdimer(M)=\frac{\prod_{e\in M}\nu_e}{\Zdimer},
\end{equation*}
where $\Zdimer=\sum_{M\in \M(\GD)}\prod_{e\in M}\nu_e$ is the {\em dimer
partition function}.

\subsection{Fisher's correspondence}

Fisher's correspondence \cite{Fisher} holds for a general Ising model defined
on a finite graph $G$ embedded on a surface without boundary, 
with coupling constants $J$. We use the following slight variation of the
correspondence.

The decorated graph, on which the dimer configurations live, is constructed from
$G$ as follows. 
Every vertex of degree $k$ of $G$ is replaced by a {\em
decoration} consisting of $3k$ vertices: a triangle is attached to
every edge incident to this vertex, and these triangles are joined by edges in
a circular way, see Figure \ref{fig:decorated_graph} below. This new graph,
denoted by~$\GD$, 
is also embedded on the surface without boundary and has vertices of degree $3$.
It is referred to as the {\em Fisher graph} of $G$. 

\begin{figure}[ht]
\begin{center}
\includegraphics[height=2.8cm]{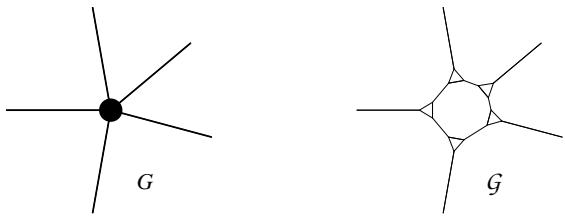}
\caption{Left: a vertex of $G$ with its incoming edges. Right: corresponding
decoration of $\GD$.}
\label{fig:decorated_graph}
\end{center}
\end{figure}
Fisher's correspondence uses the high temperature expansion of the Ising
partition function, see for example \cite{Baxter:exactly}:
\begin{equation*}\label{eq:Zhightemp}
\Zising=\left(\prod_{e\in
E(G)}\cosh(J_e)\right)2^{|V(G)|}\sum_{\C\in\P}
\prod_{e\in\C} \tanh(J_e),
\end{equation*}
where $\P$ is the family of all polygonal contours drawn on $G$, for
which every edge of $G$ is used at most once. This expansion defines a measure
on the set of polygonal 
contours $\P$ of $G$: the probability of occurrence of a polygonal
contour $\C$ is proportional to the product of the weights of the edges
it contains, where the weight of an edge $e$ is $\tanh(J_e)$.

Here comes the correspondence: to any contour configuration $\C$ coming
from the high-temperature 
expansion of the Ising model on $G$, we associate $2^{|V(G)|}$ dimer
configurations on $\GD$: edges present (resp. absent) in $\C$ are
absent (resp. present) in the corresponding dimer configuration of $\GD$. Once
the state of these edges is fixed, there is, for every decorated vertex,
exactly two ways to complete the configuration into a dimer configuration.
Figure \ref{fig:Fisher_correspondence} 
below gives an example in the case where $G$ is the square lattice $\ZZ^2$.\\

\begin{figure}[ht]
\begin{center}
\includegraphics[width=\linewidth]{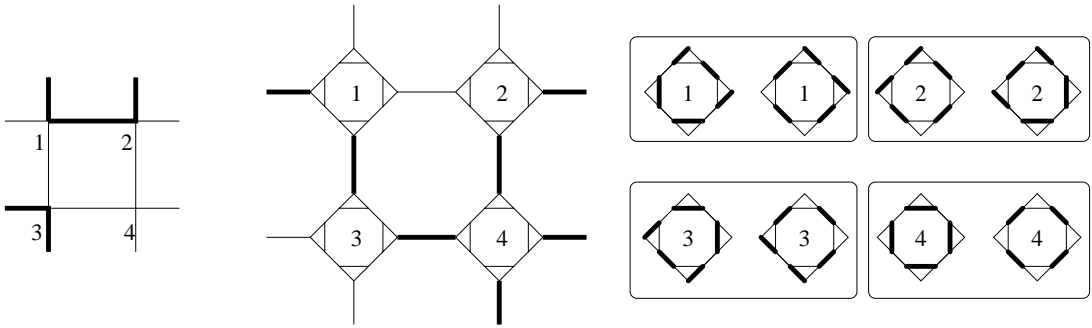}
\caption{Polygonal contour of $\ZZ^2$, and corresponding dimer configurations
  of the associated Fisher graph.}
\label{fig:Fisher_correspondence}
\end{center}
\end{figure}

Let us assign, to an edge $e$ of $\GD$, weight $\nu_e=1$, if it belongs to a
decoration; and weight $\nu_e=\coth{J_e}$, if it corresponds to an edge of $G$.
Then the correspondence is measure-preserving: every contour configuration $\C$
has the same number ($2^{|V(G)|}$) of images by this correspondence, and the
product of the 
weights of the edges in $\C$, $\prod_{e\in \C} \tanh(J_e)$
is proportional to the weight $\prod_{e\not\in\C} \coth(J_e)$  of any of
its corresponding 
dimer configurations for a proportionality factor, $\prod_{e\in E(G)}
\tanh(J_e)$, which is independent of $\C$.

As a consequence of Fisher's correspondence, we have the following relation
between the 
Ising and dimer partition functions:
\begin{equation*}
\Zising=\left(\prod_{e\in E(G)}\sinh(J_e)\right)\Zdimer.
\end{equation*}

Fisher's correspondence between Ising contour configurations and dimer configurations naturally 
extends to the case where $G$ is an infinite planar graph.

\subsection{Critical dimer model on Fisher graphs}

Consider a critical Ising model defined on an isoradial graph $G$ embedded in
the torus, or in the plane. Then, the dimer weights of the corresponding dimer
model on the Fisher graph $\GD$ are:
\begin{equation}\label{equ:dimercritical}
  \nu_e=\begin{cases}
    1 & \text{if $e$ belongs to a decoration,}\\
    \nu(\theta_e)=\cot\left(\frac{\theta_e}{2}\right) & \text{if $e$ comes from an edge of $G$.}
  \end{cases}
\end{equation}
We refer to these weights as {\em critical dimer weights}, and to the
corresponding dimer model as {\em critical dimer model 
on the Fisher graph $\GD$}.

\section{Critical characteristic
polynomials}\label{sec:characteristicpolynomials}

In this section we define the critical dimer and Laplacian characteristic
polynomials. We then state Forman's theorem proving that the Laplacian
characteristic polynomial is a generating function for CRSFs of the underlying
graph.

\subsection{Critical dimer characteristic polynomial}\label{sec:dimercharact}

The dimer model has the specific feature of having an explicit formula for
the partition function due to Kasteleyn \cite{Kast61} and independently to
Temperley and Fisher \cite{TemperleyFisher}. It involves a weighted adjacency
matrix of the underlying graph known as a \emph{Kasteleyn matrix}. Let us define
it for the critical dimer model on a Fisher graph $\GD$, which we
assume to be the Fisher graph of an infinite, $\ZZ^2$-periodic isoradial graph
$G$. Recall that $\GD_1=\GD/\ZZ^2$ denotes the fundamental domain of $\GD$.

A {\em Kasteleyn orientation} of $\GD$ is an orientation of the edges of $\GD$
such that all 
elementary cycles are {\em clockwise odd}, i.e. when
traveling clockwise around the edges of any elementary cycle of $\GD$, the
number of co-oriented edges is odd. When the graph is planar, such an 
orientation always exists \cite{Kasteleyn}. For later purposes, we need to keep
track of the orientation of the
edges of $\GD$.
We thus choose a specific Kasteleyn orientation of $\GD$ in which every 
triangle of every decoration is oriented clockwise. Having a Kasteleyn
orientation of the graph $\GD$ then amounts to finding a Kasteleyn orientation
of the planar graph obtained from $\GD$ by contracting each triangle to a single
vertex, which exists by Kasteleyn's theorem \cite{Kasteleyn}. Refer to Figure
\ref{fig:Kast_orientation} for an example of such an
orientation in the case where $G=\ZZ^2$.\\
\vspace{2cm}

\begin{figure}[ht]
\begin{center}
\includegraphics[height=5cm]{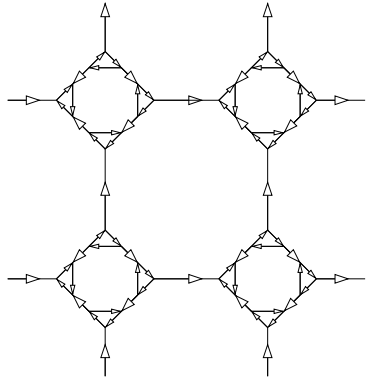}
\caption{An example of Kasteleyn orientation of the Fisher graph of $\ZZ^2$, in
which every 
triangle of every decoration is oriented clockwise.}
\label{fig:Kast_orientation}
\end{center}
\end{figure}

The {\em Kasteleyn matrix} corresponding to such an orientation is an infinite
matrix, 
whose rows and columns are indexed by vertices of $\GD$, defined by:

\begin{equation*}
K_{x,y}=\eps_{x,y}\nu_{xy},
\end{equation*}
where 
\begin{equation*}
\eps_{x,y}=
\begin{cases}
 1&\text{ if } x\sim y,\text{ and } x\rightarrow y\\
-1&\text{ if } x\sim y,\text{ and } x\leftarrow y\\
0&\text{ else},
\end{cases}
\end{equation*}
and $\nu$ is the critical dimer model weight function of Equation
\eqref{equ:dimercritical}. Note that $K$ can be interpreted as an operator
acting on $\CC^{V(\GD)}$:
\begin{equation*}
\forall f \in\CC^{V(\GD)},\quad (Kf)_x=\sum_{x\in V(\GD)}K_{x,y}f_{y}.
\end{equation*}

The {\em critical dimer characteristic polynomial}, denoted by
$\Pdimer(\zs,\ws)$, is the determinant of the Fourier transform of the 
Kasteleyn operator. More explicitly, let $\gamma_\hs$ 
and $\gamma_\vs$ be two paths in the dual graph of $\GD_1$ winding once around
the torus
horizontally and vertically 
respectively. Then, the Fourier transform of $K$ is the modified weight
Kasteleyn matrix $K_1(\zs,\ws)$ whose lines and columns are indexed by vertices
of 
$\GD_1$, and whose coefficients are those of $K$ multiplied by $\zs^{\pm 1}$ 
(resp. $\ws^{\pm 1}$) when the corresponding 
edge is crossed by the horizontal cycle (resp. vertical cycle), and the sign
$\pm 1$ is defined by the Kasteleyn orientation 
of the edge. The \emph{critical dimer characteristic polynomial} then is:
$$
\Pdimer(\zs,\ws)=\det K_1(\zs,\ws).
$$
It is the key ingredient used in explicit formulae
for the critical dimer model defined on the infinite graph $\GD$ or on the
finite toroidal graph $\GD_n=\GD/n\ZZ^2$, $n\in\NN^*$, see
\cite{BoutillierdeTiliere:iso_perio}. More precisely, the partition function of
the fundamental domain $\GD_1$ can be expressed as a linear combination of the
square root of $\Pdimer(\zs,\ws)$ evaluated at $\zs,\,\ws\in\{\pm
1\}$, and the partition function of $\GD_n$ can be expressed using
$\Pdimer(\zs,\ws)$ evaluated at $\zs^{n},\,\ws^n\in\{\pm 1\}$.

By expanding $\Pdimer(\zs,\ws)$, see also the proof of Lemma $13$ of
\cite{BoutillierdeTiliere:iso_perio},
$$
\Pdimer(\zs,\ws)=\sum_{\sigma\in\S_{|V(\GD_1)|}}\mathrm{sgn}(\sigma)\prod_{x\in
V(\GD_1)}(K_1(\zs,\ws))_{x,\sigma(x)},
$$
and using the fact that $K_1(\zs,\ws)$ is an adjacency matrix, one observes that
the only contribution to the sum comes from configurations which are unions of
disjoint cycles covering all vertices of $\GD_1$, such that trivial cycles
(homotopic to a point) are of even length (when the length is $2$, it is
then a doubled edge), non-trivial cycles (with non-trivial
homology) can be of even or odd length, and each non-trivial cycle
contributes a term $\zs^{h}\ws^{v}$ where $(h,v)$ is its
homology class. Moreover, since cycles are disjoint, non-trivial
cycles must be parallel. The difference between these configurations and
super-imposition of dimer configurations, also known as \emph{double-dimers},
lies in the terms $\zs,\,\ws$ and in the fact that non-trivial cycles can be of
odd length; double-dimer configurations can be recovered by taking a linear
combination of $\Pdimer(\zs,\ws)$ with $\zs,\ws\in\{\pm 1\}$. We refer
to configurations counted by $\Pdimer(\zs,\ws)$ as \emph{`double-dimer'
configurations}.

\subsection{Critical Laplacian characteristic polynomial}

A generalization of Kirchhoff's matrix tree theorem due to Forman \cite{Forman}
proves that the Laplacian characteristic polynomial is a generating function for
cycle rooted spanning forests, which are the natural pendent of spanning
trees when working on the torus. In this section we first define cycle rooted
spanning forests, then the Laplacian characteristic polynomial, and finally
state Forman's theorem. 

We let $G$ be an infinite, $\ZZ^2$-periodic isoradial graph,
and $G_1=G/\ZZ^2$ be the fundamental domain. Note that the content of Sections
\ref{subsec321} and \ref{subsec322} holds in more generality, i.e. when $G_1$ is
any graph embedded on the torus and $\rho$ (see below) is any positive weight
function on unoriented edges of $G_1$.

\subsubsection{Cycle rooted spanning forests}\label{subsec321}

A {\em cycle-rooted tree} (CRT) of a toroidal graph $G_1$ is a connected
subgraph of $G_1$ with a unique non-trivial cycle. A {\em cycle-rooted spanning
forest} (CRSF) is a collection of disjoint
cycle-rooted trees covering 
every vertex of $G_1$, thus implying that all non-trivial cycles are parallel.
An {\em oriented} CRT (OCRT) is a 
CRT in which edges of the branches are oriented towards the non-trivial cycle,
and the non-trivial cycle 
is oriented in one of the two possible ways. An {\em oriented} CRSF (OCRSF) is a
CRSF consisting of OCRTs.

Let us denote by $T$ a generic OCRT of $G_1$, by $F$ a generic oriented OCRSF
of $G_1$, and by
$\F(G_1)$ the collection of OCRSFs of $G_1$. 

\begin{rem}\label{rem:OCRSF}$\,$
\begin{itemize}
\item To a CRSF naturally corresponds $2^{|\text{non-trivial cycles}|}$ OCRSFs.
\item A CRSF is characterized as a subset of $|V(G_1)|$ edges of $G_1$
containing no trivial cycle.
\item An OCRSF is characterized as a subset of oriented edges of $G_1$ such
that each vertex 
has exactly one outgoing edge of this subset, and which contains no trivial cycle. 
\end{itemize}
\end{rem}

Let $\gamma_{\hs}$ and $\gamma_{\vs}$ be two paths in the dual graph of $G_1$
winding
once around the torus horizontally and vertically respectively. Assume that
$\gamma_{\hs}$ and $\gamma_{\vs}$ are assigned a reference orientation.
The {\em homology class} of an OCRT $T$, denoted by $H(T)=(h(T),v(T))$, is
defined to be the homology class of its non-trivial cycle in $\ZZ^2$. Define
the \emph{reference number} of $T$ to be: 
\begin{equation*}\label{def:H_0}
H_0(T)=(h_0(T),v_0(T))=\pm(h(T),v(T)),
\end{equation*}
where the sign is chosen so that, $h_0\geq 0$, and $v_0\geq 0$ when $h_0=0$.
Note that this definition is independent of the orientation of the non-trivial
cycle, so that it also makes sense for CRTs. Define the \emph{sign of the
non-trivial cycle} of $T$ to be:
\begin{equation*}
N(T)=\mathbf{1}_{\{H(T)=H_0(T)\}}-\mathbf{1}_{\{H(T)=-H_0(T)\}}.
\end{equation*}
Then, the homology class of the OCRT $T$ can be rewritten as:
$$
H(T)=N(T)H_0(T).
$$
Let $F$ be an OCRSF of $G_1$, and denote by $T_1,\cdots,T_n$, its tree
components, then the
homology class
$H(F)$ of $F$, is naturally defined by: 
\begin{align*}
H(F)&=(h(F),v(F))=\sum_{i=1}^n H(T_i)=\Bigl(\sum_{i=1}^n h(T_i),\sum_{i=1}^n
v(T_i)\Bigr).
\end{align*}
Since non-trivial cycles of the CRT components of $F$ are parallel, we deduce
that the number $H_0(T_i)$, $i\in\{1,\cdots,n\}$, is independent of $i$. It is
then natural to define the \emph{reference number} of the OCRSF $F$ by:
\begin{equation*}
H_0(F)=(h_0(F),v_0(F))=H_0(T_i).
\end{equation*}
As a consequence, the homology class of the OCRSF $F$ can be rewritten as:
\begin{equation*}
H(F)=N(F)H_0(F),
\end{equation*}
where $N(F)=\sum_{i=1}^n N(T_i)$ is the \emph{signed number of cycles} of the
OCRSF $F$.

\subsubsection{Critical Laplacian characteristic polynomial}\label{subsec322}

Suppose that (unoriented) edges of $G$ are 
assigned Kenyon's critical weight function for the Laplacian
\cite{Kenyon3}, denoted by $\rho$, 
$$
\forall e\in E(G),\;\rho_e=\rho(\theta_e)=\tan\theta_e,
$$
where $\theta_e$ is the rhombus half-angle of the edge $e$.
Then, the \emph{critical Laplacian} $\Delta$ on $G$, is represented by the
following matrix, also denoted $\Delta$, whose lines and columns are indexed by 
vertices of $G$:
\begin{equation*}
 \Delta_{x,y}=\left\{
\begin{array}{ll}
 \rho_{xy}&\text{ if }x\sim y\\
-\sum_{y\sim x}\rho_{xy}&\text{ if } x=y\\
0&\text{ otherwise}.
\end{array}\right.
\end{equation*}

The {\em critical Laplacian characteristic polynomial}, $\Plap(\zs,\ws)$, is
the determinant of the Fourier transform of the 
Laplacian operator, that is, $
\Plap(\zs,\ws)=\det\Delta_1(\zs,\ws),
$
where $\Delta_1(\zs,\ws)$ is the modified weight Laplacian matrix defined in a
way similar to the modified weight Kasteleyn matrix.

A remarkable fact due to Kirchhoff is that, when the graph is finite and
embedded in the plane, the absolute value of any cofactor of the Laplacian
matrix yields the weighted number of spanning trees. Forman generalized this
result, and for the case of the torus, his result can be stated as, see also
Lemma $9$ of \cite{BoutillierdeTiliere:iso_perio}:

\begin{thm}[\cite{Forman}]
The critical Laplacian characteristic polynomial is the following combinatorial
sum:
\begin{equation*}
\Plap(\zs,\ws)=\sum_{F\in \F(G_1)}\left(\prod_{e=(x,y)\in F}\rho_{xy}
\right)\prod_{T\in
F}(1-\zs^{h(T)}\ws^{v(T)}).
\end{equation*}
\end{thm}
Note that the weight function $\rho$ is independent of the orientation
of the edges.

\section{Matrix-tree theorem for the Kasteleyn matrix}\label{sec:essential}

Consider a critical Ising model defined on an infinite, $\ZZ^2$-periodic
isoradial
graph $G$, and let $\GD$ be the Fisher graph of the corresponding critical dimer
model.
In Section~\ref{sec:matrix-tree}, we state and prove
Theorem~\ref{thm:matrix-tree}, which is a matrix-tree type theorem for the
Kasteleyn matrix, thus establishing that the critical dimer characteristic
polynomial can be rewritten as a generating function for OCRSFs of $\GD_1$. In
Section \ref{sec:charac}, we analyze this polynomial and show that the
contribution of some OCRSFs cancel out, leading to the definition of
\emph{essential OCRSFs} of $\GD_1$. Then, in Section \ref{sec:statementmain}, we
give a precise statement of Theorem \ref{thm:main}, and an idea of the proof.
Section~\ref{sec:notations} is dedicated to notations.

\subsection{Notations}\label{sec:notations}

We specify notations and terminology for the isoradial graph $G$
and for the corresponding Fisher graph $\GD$, which were introduced in
\cite{BoutillierdeTiliere:iso_gen}, see also Figure
\ref{fig:notations}. These will be used throughout the
remainder of the paper.

Edges of $\GD$ corresponding to edges of $G$ are referred to as {\em long}
edges, and edges of the decorations of $\GD$ are referred to as {\em short}
ones. 

Vertices of the graph $G$ are written in boldface, and those of $\GD$ with
normal symbols. Edges and edge subsets of $G$ are also written in boldface.
Let $x$ be a vertex of $\GD$, then $x$ belongs to the decoration
corresponding to a
unique vertex $\xb$ of $G$. We shall also denote by $\xb$
the decoration in $\GD$. Conversely, vertices of a
decoration $\xb$ of $\GD$ are labeled as follows. Let
$d_{\xb}$ be the degree of the vertex $\xb$ in $G$, then the 
decoration $\xb$ of $\GD$ consists of $d_{\xb}$ \emph{triangles}, labeled
$t_1(\xb),\cdots,t_{d_{\xb}}(\xb)$ in counterclockwise order, and $d_{\xb}$ {\em
inner edges}. Vertices of the triangle $t_k(\xb)$ are labeled 
$v_k(\xb),w_k(\xb),z_k(\xb)$ in counterclockwise order, where $v_k(\xb)$ is the
only vertex incident to a long edge. We also
refer to the triangle $t_k(\xb)$ as the {\em triangle of the vertex
$v_k(\xb)$}. Whenever no confusion occurs, we drop the argument $\xb$ in the
notations above.

\begin{figure}[h]
\begin{center}
\includegraphics[height=3.3cm]{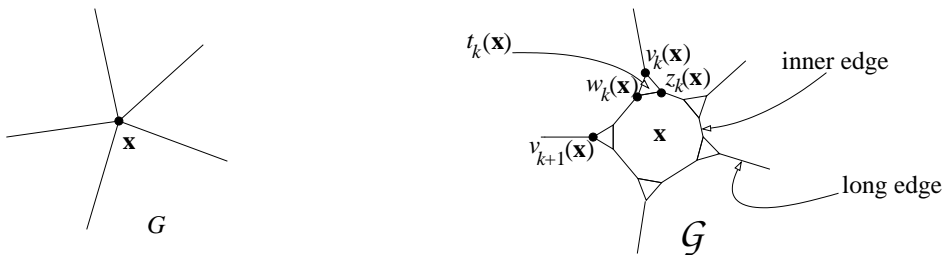}
\caption{Left: vertex $\xb$ of the graph $G$. Right: corresponding decoration
$\xb$ of~$\GD$.}
\label{fig:notations}
\end{center}
\end{figure}

Define a vertex $x$ of $\GD$ to be of {\it type `$v$'}, if
$x=v_k(\xb)$ for some vertex $\xb$ of $G$ and some $k\in\{1,\cdots,d_{\xb}\}$,
and similarly for `$w$' and `$z$'.

The isoradial embedding of the graph $G$ fixes an embedding of the
corresponding diamond graph $\GR$. There is a natural way of assigning rhombus
unit-vectors of $\GR$ to vertices of $\GD$: for every vertex $\xb$ of $G$, and
every $k\in\{1,\cdots,d_{\xb}\}$, let us associate the rhombus unit-vector
$e^{i\alpha_{w_{k}(\xb)}}$ to $w_k(\xb)$, $e^{i \alpha_{z_{k}(\xb)}}$ to
$z_k(\xb)$, and the two rhombus-unit vectors $e^{i\alpha_{w_{k}(\xb)}}$,
$e^{i\alpha_{z_{k}(\xb)}}$ to $v_k(\xb)$, as in Figure \ref{fig:rhombusvectors}
below. Note that $e^{i\alpha_{w_k(\xb)}}=e^{i\alpha_{z_{k+1}(\xb)}}$.

\begin{figure}[ht]
\begin{center}
\includegraphics[height=3.3cm]{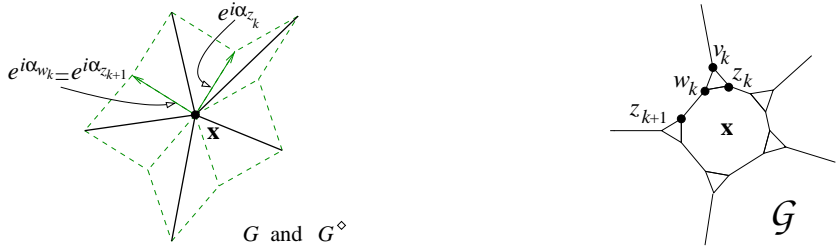}
\end{center}
\caption{Rhombus vectors of the diamond graph $\GR$ assigned to vertices of
$\GD$.}
\label{fig:rhombusvectors}
\end{figure}

\subsection{Matrix-tree theorem for the Kasteleyn matrix}\label{sec:matrix-tree}

In this section, we prove a matrix-tree type theorem for the Kasteleyn matrix
$K_1(\zs,\ws)$ of the critical dimer model on 
the graph $\GD_1$. A key requirement for such a
theorem to hold is to have a vector in the 
kernel of the matrix $K$, which is the subject of the next section.

\subsubsection{Vector in the kernel of the Kasteleyn matrix}\label{sec:421}

The vector in the kernel of the Kasteleyn matrix $K$ is naturally obtained by
setting $\lambda=0$ in the definition of the 
complex-valued function introduced in Section 4.2.2 of
\cite{BoutillierdeTiliere:iso_gen}. More precisely, we define $f=(f_x)_{\{x\in
V(\GD)\}}$, by:

\begin{equation}\label{def:f}
f_x=
\begin{cases}
e^{-i\frac{\alpha_{w_k(\xb)}}{2}}&\text{ if } x=w_k(\xb)\\
-e^{-i\frac{\alpha_{z_k(\xb)}}{2}}&\text{ if } x=z_k(\xb)\\
f_{w_k(\xb)}+f_{z_k(\xb)}&\text{ if } x=v_k(\xb),
\end{cases}
\end{equation}
for every vertex $\xb$ of $G$, and every $k\in\{1,\cdots,d_{\xb}\}$. Then,
setting $\lambda=0$ in Proposition~$15$ of \cite{BoutillierdeTiliere:iso_gen}
yields:
\begin{lem}\label{lem:kernel}
The vector $f$ is in the kernel of the matrix $K$. That is, if we let $x$ be a
vertex of $\GD$, and $x_1,x_2,x_3$ be its three neighbors, then:
\begin{equation*}
(Kf)_x=\sum_{i=1}^3 K_{x,x_i}f_{x_i}=0.  
\end{equation*}
\end{lem}

\begin{rem}
In order for the vector $f$ to be well defined, the angles $\alpha_{w_k(\xb)}$,
$\alpha_{z_k(\xb)}$ need to be well defined mod $4\pi$, indeed half-angles need
to be well defined mod $2\pi$. The latter are defined inductively in
\cite{BoutillierdeTiliere:iso_gen} as follows, see also Figure \ref{fig:angles}.
Note that the definition relies on our choice of Kasteleyn orientation of
Section \ref{sec:dimercharact}. Fix a vertex $\xb_0$
of $G$, and set $\alpha_{z_1(\xb_0)}=0$. Then, for vertices of $\GD$ in
the decoration of a vertex $\xb\in G$, define:
\begin{align}\label{eq:angle_intra_deco}
\nonumber
&\alpha_{w_k(\xb)}=\alpha_{z_k(\xb)}+2\theta_k(\xb),\text{ where
  $\theta_k(\xb)>0$ is the rhombus half-angle of Figure \ref{fig:angles},}\\ 
&\alpha_{z_{k+1}(\xb)}=
\begin{cases}
\alpha_{w_k(\xb)}& \text{if the edge } w_k(\xb) z_{k+1}(\xb) \text{ is oriented
from } w_k(\xb) \text{ to } z_{k+1}(\xb) \\
\alpha_{w_k(\xb)}+2\pi&\text{else}.
\end{cases}
\end{align}
Here is the rule defining angles in the neighboring decoration, corresponding to
a vertex $\yb$ of $G$. Let $k$ and $\l$ be indices such that $v_k(\xb)$ is
adjacent to $v_\l(\yb)$ in $\GD$. Then, define:
\begin{equation}\label{eq:angle_chg_deco}
\alpha_{w_\l(\yb)}=
\left\{
\begin{array}{ll}
\alpha_{w_k(\xb)}-\pi&\text{if the edge } v_k(\xb)v_\l(\yb) \text{ is
  oriented from } v_k(\xb)\text { to } v_\l(\yb)\\
\alpha_{w_k(\xb)}+\pi&\text{ else}.
\end{array}
\right.
\end{equation}

\begin{figure}[ht]
\begin{center}
\includegraphics[width=13cm]{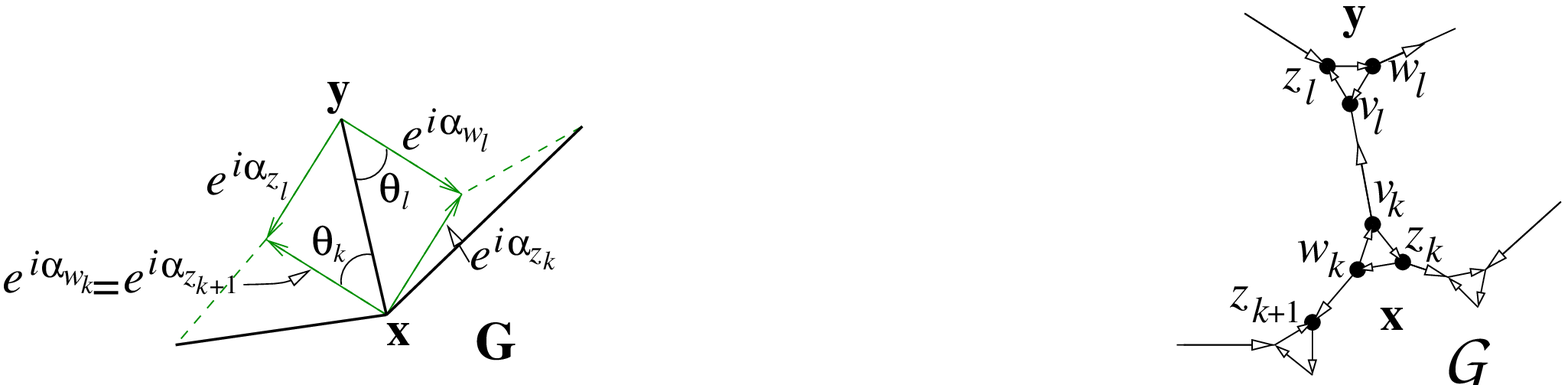}
\end{center}
\caption{Notations for the definition of the angles in $\RR/4\pi\ZZ$.}
\label{fig:angles}
\end{figure}

\begin{lem}[\rm\cite{BoutillierdeTiliere:iso_gen}]\label{lem:angles4pi}
For every vertex $\xb$ of $G$, and every $k\in\{1,\cdots,d(\xb)\}$, the
angles 
$\alpha_{w_k(\xb)}$, $\alpha_{z_k(\xb)}$, are well defined in $\RR/4\pi\ZZ$.
\end{lem}
\end{rem}

\subsubsection{Matrix-tree theorem}

The matrix-tree type theorem is most easily written for the matrix $K^0$, which
is the following gauge transformation of the Kasteleyn matrix $K$:
\begin{equation*}
K^0=D^* K D, 
\end{equation*}
where $D$ is the diagonal matrix whose elements are indexed by vertices of
$\GD$, and such that $D_{x,x}=f_x$. Let $\Pdimer^0(\zs,\ws)=\det K_1^0(\zs,\ws)$
be the characteristic polynomial of the matrix $K^0$, then we clearly have:
\begin{equation}\label{equ:30}
\Pdimer^0(\zs,\ws)=\left(\prod_{x\in V(\GD_1)}|f_x|^2\right)\Pdimer(\zs,\ws). 
\end{equation}

\begin{thm}\label{thm:matrix-tree}{\rm [Matrix-tree theorem for the Kasteleyn
matrix]}$\,$\\
The critical dimer characteristic polynomial $\Pdimer^0(\zs,\ws)$ is the
following
combinatorial sum:
\begin{equation*}
\Pdimer^0(\zs,\ws)=\sum_{F\in\F(\GD_1)}\left(\prod_{e=(x,y)\in
F}f_x\overline{f_y}K_{x,y}\right)\prod_{T\in F}(1-\zs^{h(T)}\ws^{v(T)}).
\end{equation*}
\end{thm}

In order to prove Theorem \ref{thm:matrix-tree}, we introduce two
modified incidence matrices $M$ and 
$N$ associated to the graph $\GD$, defined as follows. Rows of $M$ are indexed
by vertices of $\GD$, columns by 
unoriented edges of $\GD$, and:
\begin{equation*}
M_{x,e}=\begin{cases}
K_{x,y}&\text{ if the edge $e$ is incident to $x$, and $e=xy$}\\
0&\text{ if $e$ is not incident to $x$}.
\end{cases} 
\end{equation*}
Rows of $N$ are indexed by unoriented edges of $\GD$, columns by vertices of
$\GD$, and:
\begin{equation*}
N_{e,x}=\begin{cases}
f_x \overline{f_y} &\text{ if the edge $e$ is incident to $x$, and $e=xy$}\\
0&\text{ if $e$ is not incident to $x$}.
\end{cases} 
\end{equation*}
The next lemma relates the Kasteleyn matrix $K^0$, and the incidence matrices
$M$ and~$N$.
\begin{lem}\label{lem:incidence}
The following identity holds:
\begin{equation}\label{eq:incidence}
K^0=MN. 
\end{equation}
\end{lem}
\begin{proof}
Let $x$ and $y$ be two vertices of $\GD$. If $x$ and $y$ are at distance more
than $2$, then clearly $(MN)_{x,y}=0=K^0_{x,y}$. If $x$ and $y$ are neighbors,
we have:
\begin{align*}
(MN)_{x,y}=M_{x,e}N_{e,y}=K_{x,y} f_y \overline{f_x}=(D^*KD)_{x,y}=K^0_{x,y}.
\end{align*}
If $x=y$, then:
\begin{align*}
(MN)_{x,x}&= \sum_{e \text{ incident to } x} M_{x,e}N_{e,x}
=\sum_{y\sim x} K_{x,y} f_x\overline{f_y}\\
&=f_x\sum_{y\sim x} \overline{K_{x,y}f_y}, \quad\quad\text{ (since $K$ is
real)}\\
&=0, \quad\quad\text{ (by Lemma \ref{lem:kernel})}\\
&=K^0_{x,x}.
\end{align*}
\end{proof}

\begin{proof}[Proof of Theorem \ref{thm:matrix-tree}]
The proof is similar to that of the matrix-tree theorem on the torus, see
\cite{BoutillierdeTiliere:iso_perio} for example. Taking the
Fourier transform of Equation 
\eqref{eq:incidence} yields,
$K_1^0(\zs,\ws)=M_1(\zs^{\frac{1}{2}},\ws^{\frac{1}{2}})N_1(\zs^{\frac{1}{2}},
\ws^ { \frac{1}{2}}),$ where $M_1(\zs,\ws)$ and $N_1(\zs,\ws)$ are the Fourier
transform of the matrices $M$ and $N$.
Moreover, since for $(\zs,\ws)\in\TT^2$ the matrix $K_1^0(\zs,\ws)$ is
skew-hermitian, we have:
\begin{equation*}
\Pdimer^0(\zs,\ws)=\det\left(M_1(\zs^{-\frac{1}{2}},\ws^{-\frac{1}{2}})N_1(\zs^{
-\frac{1}{2}},\ws^{-\frac{1}{2}})\right),
\end{equation*}
where we also use the fact that $|V(\GD_1)|=6|E(G_1)|$ is
even. We now use Cauchy-Binet's formula to expand $\Pdimer^0(\zs,\ws)$:
\begin{equation}\label{eq:Cauchy_Binet}
\Pdimer^0(\zs,\ws)=\sum_{\scriptstyle\begin{array}{c}
                S\subset E(\GD_1)\\
                |S|=|V(\GD_1)|
               \end{array}
}\det(M_1(\zs^{-\frac{1}{2}},\ws^{-\frac{1}{2}})^S)
\det(N_1(\zs^{-\frac{1}{2}},\ws^{-\frac{1}{2}})_S). 
\end{equation}
Recall that $\gamma_{\hs}$, $\gamma_{\vs}$ are two paths in the dual graph of
$\GD_1$
winding once around the torus horizontally and vertically respectively.
Suppose first that $S$ contains a trivial cycle, denoted by $C$, which
does not cross the horizontal or the vertical cycle $\gamma_{\hs}$,
$\gamma_{\vs}$. Let
us show that in this case,  
$\det(M_1(\zs^{-\frac{1}{2}},\ws^{-\frac{1}{2}})^S)=~0$. Define
$\vphi:V(\GD_1)\rightarrow\CC$ by:
\begin{equation*}
\vphi_x=
\begin{cases}
1&\text{ if $x\in C$}\\
0&\text{ else}.
\end{cases}
\end{equation*}
Then, for every unoriented edge $e\in S$, we have:
\begin{align*}
(\vphi^t &M_1(\zs^{-\frac{1}{2}},\ws^{-\frac{1}{2}})^S)_{e}=\sum_{x\in
V(\GD_1)}\vphi_x M_1(\zs^{-\frac{1}{2}},\ws^{-\frac{1}{2}})_{x,e}\\
&=
\begin{cases}
0& \text{ if } e\notin C\\
\vphi_x M_1(\zs^{-\frac{1}{2}},\ws^{-\frac{1}{2}})_{x,e}+\vphi_y
M_1(\zs^{-\frac{1}{2}},\ws^{-\frac{1}{2}})_{y,e}=K_{x,y}+K_{y,x}=0&\text{ if
}e=xy\in C.
\end{cases}
\end{align*}
Thus, $\det(M_1(\zs^{-\frac{1}{2}},\ws^{-\frac{1}{2}})^S)=0$. If $C$ crosses the
horizontal and/or vertical cycles $\gamma_{\hs}$, $\gamma_{\vs}$, then the
vector
$\vphi$ in the kernel of $(M_1(\zs^{-\frac{1}{2}},\ws^{-\frac{1}{2}})^S)^t$ can
be defined in a similar way. Therefore, the only contribution to the sum
\eqref{eq:Cauchy_Binet} comes from graphs whose number of edges equals the
number of vertices and which contain no trivial cycle, {\em i.e.} from CRSFs.

Let us now compute the contribution of a CRSF $F$. After a possible reordering
of the rows and columns of $M_1(\zs^{-\frac{1}{2}},\ws^{-\frac{1}{2}})$ and
$N_1(\zs^{-\frac{1}{2}},\ws^{-\frac{1}{2}})$, we can suppose that both these
matrices are block diagonal, each diagonal block corresponding to a connected
component of $F$, {\em i.e.} a cycle rooted tree. 

The determinant of a CRT $T$ in $M_1(\zs^{-\frac{1}{2}},\ws^{-\frac{1}{2}})$
(resp. $N_1(\zs^{-\frac{1}{2}},\ws^{-\frac{1}{2}})$) can be evaluated by
expanding it along columns (resp. rows) corresponding to leaves of the CRT.
What remains then is the evaluation of the determinant reduced to the cycle.
More precisely, suppose that edges of the branches are oriented from the leaves
to the non-trivial cycle. Then, the contribution of the branches to
$\Pdimer^0(\zs,\ws)$, is: 
\begin{equation*}
\prod_{(x,y)\in \text{ branch of }T}
K_{x,y}f_x\overline{f_y}.
\end{equation*}
Recall the definition of the reference number $H_0(T)=(h_0(T),v_0(T))$ of the
CRT $T$, given in Section \ref{subsec321}, and let
$x_1,\cdots,x_n$, be a labeling
of the vertices of its non-trivial cycle in the direction given by $H_0(T)$.
Then, the contribution of $T$ to $M_1(\zs^{-\frac{1}{2}},\ws^{-\frac{1}{2}})$
is:
\begin{equation*}
\left(\prod_{i=1}^n
K_{x_i,x_{i+1}}\right)\zs^{-\frac{h_0(T)}{2}}\ws^{-\frac{v_0(T)}{2}}+(-1)^{n+1}
\left(\prod_{i=1}^n
K_{x_{i+1},x_i}\right)\zs^{\frac{h_0(T)}{2}}\ws^{\frac{v_0(T)}{2}},
\end{equation*}
and the contribution to $N_1(\zs^{-\frac{1}{2}},\ws^{-\frac{1}{2}})$ is:
\begin{equation*}
\left(\prod_{i=1}^n
f_{x_i}\overline{f_{x_{i+1}}}\right)\zs^{\frac{h_0(T)}{2}}\ws^{\frac{v_0(T)}{2}}
+(-1)^
{n+1}\left(\prod_{i=1}^n
\overline{f_{x_i}}f_{x_{i+1}}\right)\zs^{-\frac{h_0(T)}{2}}\ws^{-\frac{v_0(T)}{2
}} . 
\end{equation*}
Since the Kasteleyn matrix $K$ is skew-symmetric, we have $(-1)^{n+1}
\prod_{i=1}^{n} K_{x_{i+1},x_i}=-\prod_{i=1}^n K_{x_i,x_{i+1}}$, 
thus the contribution of the CRT $T$ is:
\begin{equation*}
\left(\prod_{i=1}^n
f_{x_i}\overline{f_{x_{i+1}}}K_{x_i,x_{i+1}}\right)(1-\zs^{h_0(T)}\ws^{v_0(T)})
+\left(\prod_{i=1}^n
f_{x_{i+1}}\overline{f_{x_i}}K_{x_{i+1},x_i}\right)(1-\zs^{-h_0(T)}\ws^{-v_0(T)}
).
\end{equation*}
Let $T_1$ and $T_2$ be the two OCRTs corresponding to $T$, then the contribution
of the CRT~$T$ can be rewritten as:
\begin{equation*}
\sum_{i=1}^2 \left(\prod_{e=(x,y)\in
T_i}f_x\overline{f_y}K_{x,y}\right)\left(1-\zs^{h(T_i)}\ws^{v(T_i)}\right). 
\end{equation*}
Taking the product over all CRTs of the CRSF $F$, and summing over all CRSFs
yields:
\begin{align*}
\Pdimer^0(\zs,\ws)&=\sum_{F\in\{\text{CRSF of }\GD_1\}} \prod_{T\in\{\text{CRT
of F}\}}\sum_{i=1}^2 \left(\prod_{e=(x,y)\in
T_i}f_x\overline{f_y}K_{x,y}\right)\left(1-\zs^{h(T_i)}\ws^{v(T_i)}\right),\\
&= \sum_{F\in \F(\GD_1)}\prod_{T\in F}\left(\prod_{e=(x,y)\in
T}f_x\overline{f_y}K_{x,y}\right)\left(1-\zs^{h(T)}\ws^{v(T)}\right).
\end{align*}
\end{proof}


\subsection{Characterization of OCRSFs contributing to
$\Pdimer^0(\zs,\ws)$}\label{sec:charac}

In this section, we characterize OCRSFs of $\GD_1$ which contribute to
$\Pdimer^0(\zs,\ws)$. Indeed, it turns out that the contributions of some of
them cancel out. 

More precisely, let $\Lb$ be a subset of oriented edges of
$G_1$, defining a subset of oriented long edges $\Lb$ of $\GD_1$. An OCRSF $F$
of $\GD_1$ is said to be {\em compatible with $\Lb$}, if the long edges of $F$
are exactly those of $\Lb$. We first characterize OCRSFs
compatible with $\Lb$, and then OCRSFs compatible with $\Lb$ which actually
contribute to $\Pdimer^0(\zs,\ws)$. 

\subsubsection{OCRSFs compatible with $\Lb$}\label{sec:OCRSFdecograph}

An oriented edge $(\xb,\yb)$ of $\Lb$ corresponds to a unique oriented edge
$(v_k(\xb),v_\l(\yb))$ of~$\GD_1$. We refer to the vertex $v_k(\xb)$ as a {\em
root vertex} of $\Lb$. This defines for every decoration $\xb$, a set of root
vertices denoted by $R_{\xb}(\Lb)$, and a set of non-root vertices of type
`$v$', $R_{\xb}(\Lb)^c:=\{v_1(\xb),\cdots,v_{d_{\xb}}(\xb)\}\setminus
R_{\xb}(\Lb)$, see Figure \ref{fig:fig20} below for an example. 

\begin{figure}[ht]
\begin{center}
\includegraphics[width=3.7cm]{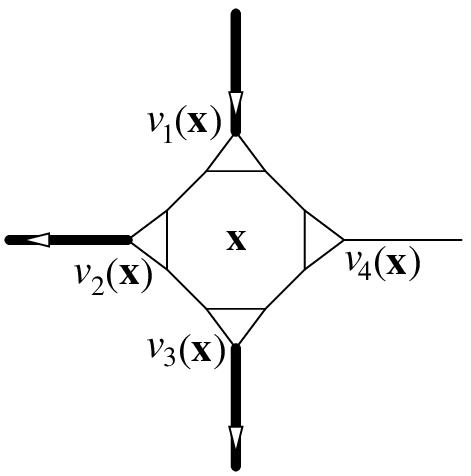}
\caption{In this example $\mathbf{L}$ is drawn in bold lines,
$R_{\xb}(\Lb)=\{v_2(\xb),v_3(\xb)\}$ and
$R_{\xb}(\Lb)^c~=~\{v_1(\xb),v_4(\xb)\}$.}
\label{fig:fig20}
\end{center}
\end{figure}

\begin{rem}
Suppose that one of the decoration $\xb$ of $\GD_1$ contains no 
root vertex of~$\Lb$, and suppose that there is an OCRSF $F$ of $\GD_1$
compatible with $\Lb$. 
Then, since by Remark~\ref{rem:OCRSF}, the OCRSF $F$ has exactly one outgoing
edge at every vertex, 
this implies that $F$ restricted to the decoration $\xb$ has a number of edges
equal to the number of 
vertices of the decoration, thus it must contain a trivial cycle and cannot be
an OCRSF. 
As a consequence, we only consider subset of oriented edges $\Lb$ of $G_1$ such
that:
\begin{center}
Every decoration of $\GD_1$ has at least one root vertex of
$\Lb$.\hspace{1cm}$(\ast)$
\end{center}
\end{rem}
 
\begin{lem}\label{lem:1}
A subset of oriented edges $F$ of $\GD_1$ is an OCRSF compatible with $\Lb$,
iff the following conditions hold:
\begin{enumerate}
\item long edges of $F$ are exactly those of $\Lb$,
\item $F$ contains no cycle consisting of long and short edges,
\item the restriction of $F$ to every decoration $\xb$, is an
oriented spanning forest with $|R_{\xb}(\Lb)|$ connected components whose roots
are distinct elements of~$R_{\xb}(\Lb)$.
\end{enumerate}
\end{lem}

\begin{rem}\label{rem:3}
Condition $3.$ is equivalent to Condition $3'$: the restriction of $F$ to every
decoration $\xb$ contains no trivial cycle, and is such that vertices
of $R_\xb(\Lb)$ have no outgoing edge of $F$, and every other vertex has exactly
one outgoing edge of $F$.
\end{rem}

\begin{proof}
Let $F$ be a subset of oriented edges of $\GD_1$. Then by the
geometry of the graph $\GD_1$, $F$ cannot have trivial cycles consisting of long
edges only. Thus, Conditions~$1,2,3'$ are equivalent to saying that the oriented
edge configuration $F$ is compatible with $\Lb$, contains no trivial cycle, and
is such that every vertex of $\GD_1$ has exactly one outgoing edge of $F$. By
Remark \ref{rem:OCRSF}, this is equivalent to saying that $F$ is an OCRSF
compatible with $\Lb$.
\end{proof}

\subsubsection{Restriction to decorations of OCRSFs contributing to
$\Pdimer^0(\zs,\ws)$}\label{sec:restrictedOCRSFdecograph}

In this section, we characterize the restriction to decorations of OCRSFs of
$\GD_1$ compatible with $\Lb$, which contribute to $\Pdimer^0(\zs,\ws)$. 
To this purpose we first introduce the following definition, see also Figure
\ref{fig:fig19}.

\begin{defi}\label{def:restricted}
A subset of oriented edges of the decoration $\xb$, is called an
\emph{essential configuration compatible with $\mathbf{L}$, of type \emph{cw} 
(resp. \emph{cclw})}, if it consists of:
\begin{enumerate}
\item all inner edges oriented clockwise (resp.
counterclockwise).
\item one of the three following $2$-edge configurations at the triangle of
every non-root vertex $v_i\in R_{\xb}(\Lb)^c$.
\begin{align*}
-
&\{(w_i,v_i),(v_i,z_i)\},\,\{(w_i,z_i),(v_i,z_i)\},\,\{(v_i,w_i),(w_i,z_i)\}
\text{ in the \emph{cw} case},\\
-
&\{(z_i,v_i),(v_i,w_i)\},\,\{(v_i,z_i),(z_i,w_i)\},\,\{(z_i,w_i),(v_i,w_i)\}
\text{ in the \emph{cclw} case},
\end{align*}
\item one of the two following $1$-edge configurations at the triangle of every
root
vertex $v_i\in R_{\xb}(\Lb)$: 
\begin{align*}
- & \{(w_i,v_i)\},\{(w_i,z_i)\} \text{ in the \emph{cw} case},\\
- &\{(z_i,v_i),\{(z_i,w_i)\}\text{ in the \emph{cclw} case},&
\end{align*}
with the additional constraint that the triangle of at least one root vertex
contains the configuration $(w_i,v_i)$ in the \emph{cw} case, and the
configuration $(z_i,v_i)$ in the \emph{cclw} case.
\end{enumerate}

\begin{figure}[ht]
\begin{center}
\includegraphics[width=7cm]{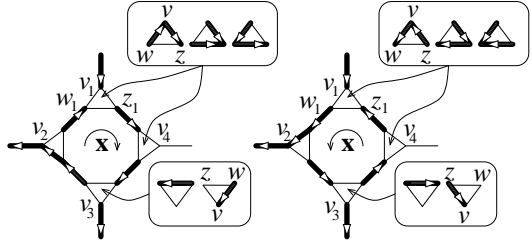}
\caption{Essential configuration of the decoration $\xb$ of type {\em cw}
(left) and of type {\em cclw} (right).}
\label{fig:fig19}
\end{center}
\end{figure}

When the type is not specified, we refer to the above as an \emph{essential
configuration compatible with $\Lb$}, and when 
no confusion occurs, we omit the specification \emph{compatible with~$\Lb$}.
\end{defi}

\begin{prop}\label{prop:contributing}$\,$
\begin{enumerate}
\item For every decoration $\xb$, an essential configuration
compatible with $\Lb$ is an oriented spanning forest 
with $|R_{\xb}(\Lb)|$ connected components whose roots are distinct elements of
$R_{\xb}(\Lb)$.
\item Let $F$ be an OCRSF of $\GD_1$ compatible with $\Lb$, which contributes to
$\Pdimer^0(\zs,\ws)$, then the restriction of 
$F$ to every decoration $\xb$, is an essential configuration compatible with
$\Lb$.
\end{enumerate}
\end{prop}

\begin{proof}
Proposition \ref{prop:contributing} is a direct consequence of the next two
lemmas.
\end{proof}

\begin{lem}\label{lem:deco}
Let $F$ be an OCRSF of $\GD_1$ compatible with $\Lb$, which contributes to
$\Pdimer^0(\zs,\ws)$. Then 
the restriction of $F$ to every decoration $\xb$ contains all inner edges of
the decoration, one edge at triangles of 
root vertices, and two edges at triangles of non-root vertices.
\end{lem}

\begin{proof}
Let $F$ be an OCRSF of $\GD_1$ compatible with $\Lb$. By Lemma \ref{lem:1}, the
restriction of $F$ to every decoration $\xb$, 
is an oriented spanning forest with $|R_{\xb}(\Lb)|$ components,
whose roots are distinct elements of $R_{\xb}(\Lb)$. 
This implies that:
\begin{enumerate}
\item[$a.$] Triangles of decorations contain at most $2$ edges of $F$.
\item[$b.$] Each vertex of $\GD_1$ has a unique outgoing edge of $F$, and each
root vertex of $R_{\xb}(\Lb)$ has as outgoing edge of $F$ the unique incident
long edge (which belongs to $\Lb$ by definition of root vertices).
\item[$c.$] The restriction of $F$ to every decoration $\xb$ contains
$3d_{\xb}-|R_{\xb}(\Lb)|$ edges.
\end{enumerate}
Let us first prove that the contribution to $\Pdimer^0(\zs,\ws)$ of OCRSFs
having two edges at triangles of root vertices cancel out. Let $\xb$ be a
decoration of $\GD_1$, and $v_i\in R_{\xb}(\Lb)$ be a root vertex of $\xb$.
Then, there are three possible oriented $2$-edge configurations at the triangle
$t_i$ of the vertex $v_i$, which satisfy the necessary requirement of Point $b.$
above, see Figure \ref{fig:fig1}:
\begin{equation*}
T_1=\{(w_i,v_i),(z_i,v_i)\},\quad T_2=\{(w_i,z_i),(z_i,v_i)\},\quad
T_3=\{(z_i,w_i),(w_i,v_i)\}.
\end{equation*}

\begin{figure}[ht]
\begin{center}
\includegraphics[width=5cm]{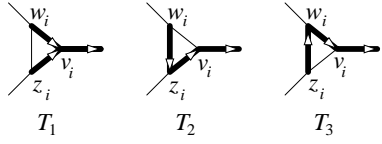}
\caption{The three possible oriented $2$-edge configurations at triangles of
root vertices.}
\label{fig:fig1}
\end{center}
\end{figure}

Suppose that there exists an OCRSF $F^1$ compatible with $\Lb$, having the edge
configuration $T_1$ at the triangle $t_i$.
 Let $F^2$ (resp. $F^3$) be the OCRSF $F^1$ modified so as to have the edge
configuration $T_2$ (resp. $T_3$) at the triangle $t_i$. Then, it is
straightforward to check that $F^1$, $F^2$ and $F^3$ are OCRSFs compatible with
$\Lb$, with oriented 
non-trivial cycles in bijection (one way to prove this is to 
use Part $1$ of Lemma
\ref{lem:replacing1edge}). By Theorem~\ref{thm:matrix-tree}, the 
contribution $C(F^1,F^2,F^3)$ of $F^1,F^2,F^3$ to $\Pdimer^0(\zs,\ws)$ is:
\begin{equation*}
C(F^1,F^2,F^3):=\sum_{j=1}^3 \left(\prod_{e=(x,y)\in
F^j}f_x\overline{f_y}K_{x,y}\right)\prod_{T\in F^j}(1-z^{h(T)}w^{v(T)}).
\end{equation*}
Since $F^1,F^2,F^3$ have oriented non-trivial cycles in bijection, the term
$\prod_{T\in F^j}(1-z^{h(T)}w^{v(T)})$ is independent of $j$. Denote by $C$ the
contribution of edges of $F^1$ which are not edges of the triangle $t_i$.
Then: 
\begin{multline*}
C(F^1,F^2,F^3)= C \prod_{T\in F^1}(1-z^{h(T)}w^{v(T)})\cdot
\\
\cdot\left(
f_{w_i}\overline{f_{v_i}}K_{w_i,v_i}f_{z_i}\overline{f_{v_i}}K_{z_i,v_i}
+
f_{w_i}\overline{f_{z_i}}K_{w_i,z_i}f_{z_i}\overline{f_{v_i}}K_{z_i,v_i}
+
f_{z_i}\overline{f_{w_i}}K_{z_i,w_i}f_{w_i}\overline{f_{v_i}}K_{w_i,v_i}
\right).
\end{multline*}

Recalling that by our choice of Kasteleyn orientation, edges of the triangles
are oriented clockwise, we have that $C(F^1,F^2,F^3)$ is equal to:
\begin{align*}
&C \prod_{T\in F^1}(1-z^{h(T)}w^{v(T)})
\left[
f_{w_i}\overline{f_{v_i}}f_{z_i}\overline{f_{v_i}}(-1)+
(-1)f_{w_i}\overline{f_{z_i}}(-1)f_{z_i}\overline{f_{v_i}}+
f_{z_i}\overline{f_{w_i}}f_{w_i}\overline{f_{v_i}}
\right]\\
&=C \prod_{T\in F^1}(1-z^{h(T)}w^{v(T)})\left[
f_{z_i}f_{w_i}\overline{f_{v_i}}(-\overline{f_{v_i}}+\overline{f_{z_i}}
+\overline{f_{w_i}})\right]\\
&=0\quad\text{(by definition of the vector $f$)}.
\end{align*}

Since this holds for every decoration $\xb$, and every root vertex
$v_i\in R_{\xb}(\Lb)$, we deduce that the contribution to $\Pdimer^0(\zs,\ws)$
of OCRSFs compatible with $\Lb$, having two edges at triangles of root vertices,
cancel out. 

As a consequence, if an OCRSF compatible with $\Lb$ contributes to
$\Pdimer^0(\zs,\ws)$, then the restriction of $F$ at the 
decoration $\xb$ contains at most, all inner edges of the decoration, one edge
at triangles of root vertices, and two edges at 
triangles of non-root vertices; that is it contains at most
$d_{\xb}+|R_{\xb}(\Lb)|+2|R_{\xb}(\Lb)^c|=3d_{\xb}-|R_{\xb}(\Lb)|$ 
edges. Since by Point $c.$ above, it must contain exactly
$3d_{\xb}-|R_{\xb}(\Lb)|$ edges, we deduce that all constraints must 
be met.
\end{proof}

\begin{lem}\label{lem:2}
A subset of edges of the decoration $\xb$, is an oriented spanning
forest with $|R_{\xb}(\Lb)|$ components whose 
roots are distinct elements of $R_{\xb}(\Lb)$, satisfying the constraints of
Lemma \ref{lem:deco}, iff it is an essential
configuration compatible with $\Lb$.
\end{lem}

\begin{proof}
By Remark \ref{rem:3}, Lemma \ref{lem:2} amounts to proving that an oriented
subset of the decoration $\xb$ contains all inner 
edges of the decoration, one edge at triangles of root vertices, two edges at
triangles of non-root vertices, and satisfies 
Condition $3'$, iff it is an essential configuration compatible with $\Lb$. This
is clear once we return to Definition~\ref{def:restricted}.
\end{proof}

\subsubsection{Essential OCRSFs of $\GD_1$}

Sections \ref{sec:OCRSFdecograph} and \ref{sec:restrictedOCRSFdecograph} above
naturally yield the 
following definitions.

\begin{defi}
Let $\tau\in\{cw,ccwl\}^{V(G_1)}$ be an assignment of type \emph{cw} or
\emph{cclw} to the vertices of $G_1$. 
Define the set of \emph{essential OCRSFs of $\GD_1$ compatible with $\Lb$ and
$\tau$}, to be the set of OCRSFs compatible with $\Lb$ whose
restriction to every decoration $\xb$, is an essential
configuration of type $\tau(\xb)$. We denote this set by
$\F^{\tau,\Lb}(\GD_1)$.
\end{defi}

As a consequence of Lemma \ref{lem:1} and Proposition \ref{prop:contributing},
we have the following characterization of  $\F^{\tau,\Lb}(\GD_1)$.

\begin{cor}\label{rem:4}$\,$
A subset of oriented edges of $\GD_1$ is an essential OCRSF compatible with
$\Lb$ and $\tau$, iff:
\begin{enumerate}
\item long edges of $F$ are exactly those of $\Lb$,
\item $F$ contains no trivial cycle consisting of short and long edges,
\item the restriction of $F$ to every decoration $\xb$ of $\GD_1$, is an
essential
configuration compatible with $\Lb$, of type $\tau(\xb)$.
\end{enumerate}
\end{cor}

\begin{defi}
Define the set of \emph{essential OCRSFs of $\GD_1$ compatible with $\Lb$},
denoted $\F^{\Lb}(\GD_1)$, by:
\begin{equation*}
\F^{\Lb}(\GD_1)=\bigcup_{\tau\in\{cw,cclw\}^{V(G_1)}}\F^{\tau,\Lb}(\GD_1), 
\end{equation*}
and the set of \emph{essential OCRSFs of $\GD_1$}, denoted $\F^0(\GD_1)$, by:
\begin{equation*}
\F^0(\GD_1)=\bigcup_{\Lb\in\boldsymbol{\mathcal{L}}} \F^{\Lb}(\GD_1),
\end{equation*}
where $\boldsymbol{\mathcal{L}}$ is the set of oriented edge configurations of
$G_1$ 
satisfying $(\ast)$. 

\end{defi}

\begin{cor}
\begin{equation*}
\Pdimer^0(\zs,\ws)=\sum_{F\in\F^0(\GD_1)}
\left(\prod_{e=(x,y)\in F}f_x\overline{f_y}K_{x,y}\right)\prod_{T\in F}
(1-\zs^{h(T)}\ws^{v(T)}).
\end{equation*}
\end{cor}
\begin{proof}
Write, $\F(\GD_1)=\F^0(\GD_1)\cup (\F^0(\GD_1))^c$ in the formula for
$\Pdimer^0(\zs,\ws)$ given by Theorem 
\ref{thm:matrix-tree}. Then use Point $2.$ of Proposition
\ref{prop:contributing}, to deduce that
the contribution of OCRSFs of $(\F^0(\GD_1))^c$ cancel out in
$\Pdimer^0(\zs,\ws)$.
\end{proof}

\subsection{Statement of main result}\label{sec:statementmain}

We can now give a precise statement of Theorem
\ref{thm:main}.
\begin{thm}\label{thm:main2}
Consider a critical Ising model defined on an infinite, $\ZZ^2$-periodic
isoradial graph $G$. Then, one can explicitly construct essential
OCRSFs of $\GD_1$ counted by the critical dimer
characteristic polynomial $\Pdimer^0(\zs,\ws)$, from OCRSFs of $G_1$ counted by
the critical Laplacian characteristic polynomial $\Plap(\zs,\ws)$.
\end{thm}

Let us give an idea of the construction. First observe that the number of
OCRSFs of $\GD_1$ is much greater than the number of OCRSFs of $G_1$ so that
there is certainly no one-to-one mapping between these set of configurations.
Rather, to every OCRSF of $G_1$, we
assign a family of OCRSFs of $\GD_1$, which have the same reference number and
the the same signed number of cycles, and such that the sum of the weights of
OCRSFs in this family is equal to the weight of the original OCRSF of $G_1$.
The family of OCRSF of $\GD_1$ is constructed by successively adding long
edges, and this construction is done by induction on the number of long edges
added, thus allowing us to keep control over properties of OCRSFs in this
family. To be more precise, we actually need to work on the graph $G_1$ and its
dual graph $G_1^*$ at the same time.
The operations used to add long edges are called \emph{licit
primal/dual} moves, and are valid in a context more general than isoradial and
Fisher graphs, so that we dedicate it Section
\ref{sec:general}. The proof of Theorem \ref{thm:main2} is the subject of
Section \ref{sec:CORRES}. In Section \ref{sec:explicitcorrespondence}, we
construct a family of CRSFs of $\GD_1$ from each CRSF of $G_1$. In
Theorem \ref{thm:correspondence} we prove that we exactly obtain
all CRSFs of $\GD_1$ counted by the dimer characteristic polynomial; and in
Section \ref{sec:weightpreserving}, we show that this construction is weight
preserving.

\section{Primal/dual edge moves on pairs of dual OCRSFs}\label{sec:general}

Definitions and results of this section are valid in a context more general
than isoradial and Fisher graphs. But, in order not to introduce too many
notations, we let $G_1$ be any graph embedded in the torus.

In Section \ref{sec:defiOCRSFpol}, we first define a general \emph{OCRSF
characteristic polynomial}, allowing for weights which depend on oriented
edges, and we prove a useful rewriting of this polynomial, using
pairs of dual OCRSFs. In Section
\ref{sec:replacing2edges}, we define and prove properties of \emph{licit
primal/dual edge moves}, which are edge
moves performed on pairs of dual OCRSFs, and are one of the key
ingredients of the construction of Theorem \ref{thm:main2}. Licit primal/dual
edge moves rely on a natural edge operation
performed on one OCRSF only, which is the
subject of Section \ref{sec:replacingedges}.

\subsection{OCRSF characteristic polynomial}\label{sec:defiOCRSFpol}

We use definitions and notations introduced in Section
\ref{subsec321}. Suppose that a complex
weight function $\rho$ is assigned to oriented edges of $G_1$, that is every
oriented edge $(x,y)$ has a weight $\rho_{(x,y)}\in\CC$. The
{\em OCRSF characteristic polynomial}, corresponding to the 
weight function $\rho$, denoted $\POCRSF(\zs,\ws)$, is defined by:
\begin{equation}\label{equ:OCRSFpoly}
\POCRSF(\zs,\ws)= \sum_{F\in\F(G_1)} \rho(F)\prod_{T\in
F}(1-\zs^{h(T)}\ws^{v(T)}),
\end{equation}
where $\rho(F)=\prod_{(x,y)\in F}\rho_{(x,y)}$.
\begin{rem}
When $\rho$ is a positive weight function on unoriented edges of $G_1$, then the
OCRSF characteristic polynomial is simply the
Laplacian characteristic polynomial.
\end{rem}
We now prove a useful rewriting of $\POCRSF(\zs,\ws)$. In order to do
this, we need a few more facts and definitions.
Denote by $G_1^*$ the dual graph of $G_1$ and let us, for a moment, consider
edge configurations as unoriented. 
It is a general fact that if $F$ is a CRSF of $G_1$, then the complementary of
the dual edge configuration, consisting exactly of 
the dual edges of the edges absent in $F$, is a CRSF of $G_1^*$, with non
trivial cycles parallel to those of $F$,
such that primal and dual non trivial cycles alternate along the torus. It is
referred to as the \emph{dual CRSF of $F$}.

Let $F$ and $F^*$ be OCRSFs of $G_1$ and $G_1^*$ respectively, then $F$ and
$F^*$ are said to be {\em dual} of eachother, 
if their unoriented versions are. This means that to a pair of dual CRSFs corresponds 
$4^{|\scriptstyle \text{non-trivial cycles}|}$ pairs of
dual OCRSFs. We denote by $\F(G_1,G_1^*)$ the set of pairs of dual OCRSFs,
that is:
$$
\F(G_1,G_1^*)=\{(F,F^*)\in\F(G_1)\times\F(G_1^*)\,:\,F,\,F^*\text{ are dual
OCRSFs}\}.
$$

Let $(F,F^*)$ be a pair of dual OCRSFs of $G_1$ and $G_1^*$. 
Recall that the homology class $H(F)$ of $F$ can be rewritten as:
$$
H(F)=N(F)H_0(F),
$$
where $H_0(F)=(h_0(F),v_0(F))$ is the reference number of $F$, and $N(F)$ is the
signed number of non-trivial cycles of $F$ defined in Section \ref{subsec321};
and similarly for the homology class of $F^*$. Observe that $F$ and $F^*$ have
the same reference number, so that it makes sense to refer to $H_0(F)$ as the
\emph{reference number} of the pair $(F,F^*)$. Let us denote by $N(F,F^*)$ the
sum $N(F)+N(F^*)$. Then, we have the following
rewriting of $\POCRSF$.

\begin{lem}\label{lem:rewriting}
The OCRSF characteristic polynomial of $G_1$ can be rewritten as:
\begin{equation*}
\POCRSF(\zs,\ws)=\sum_{
(F,F^*)\in \F(G_1,G_1^*)}
\rho(F)\,
(-\zs^{h_0(F)}\ws^{v_0(F)})^{\frac{1}{2}N(F,F^*)}.
\end{equation*}
\end{lem}

\begin{proof}
Denote by $T_1,\cdots,T_n$ the OCRT components of a generic OCRSF $F$ of $G_1$.
By expanding the product of Equation \eqref{equ:OCRSFpoly}, the polynomial
$\POCRSF(\zs,\ws)$ can be rewritten as:
\begin{align*}\label{equ:proof_rewriting1}
\POCRSF(\zs,\ws)&=\sum_{F\in \F(G_1)}\rho(F)\prod_{i=1}^n 
\sum_{\eps_{i}\in\{0,1\}}\left(-\zs^{h(T_i)}\ws^{v(T_i)}\right)^{\eps_{i}}
\nonumber\\
&=\sum_{F\in
\F(G_1)}\rho(F)\sum_{(\eps_1,\cdots,\eps_n)\in\{0,1\}^n}\prod_{i=1}^n 
\left(-\zs^{h(T_i)}\ws^{v(T_i)}\right)^{\eps_{i}}.
\end{align*}

Fix an OCRSF $F$ of $G_1$, and let $(\eps_1,\cdots,\eps_n)\in\{0,1\}^n$.
Then to $F$ assign the following dual OCRSF $F^*$: 
if $\eps_i=1$, take $T_i^*$ to be co-oriented with $T_i$, and if $\eps_i=0$, take $T_i^*$ to be contra-oriented. Thus, 
there is a one-to-one correspondence between dual OCRSFs of $F$, and sequences $(\eps_1,\cdots,\eps_n)\in\{0,1\}^n$, and:
\begin{align*}
&\eps_{i}=\pm\frac{1}{2}(N(T_i)+N(T_i^*)),\\
&\eps_{i}(h(T_i),v(T_i))=\frac{1}{2}\left(h(T_i)+h(T_i^*),v(T_i)+v(T_i^*)\right).
\end{align*}
This implies that:
\begin{align*}
(-1)^{\eps_i}&=(-1)^{\frac{1}{2}(N(T_i)+N(T_i^*))},\\
\eps_{i}(h(T_i),v(T_i))&=\frac{1}{2}(N(T_i)+N(T_i^*))(h_0(T_i),v_0(T_i)).
\end{align*}
The proof is concluded by recalling that $\sum_{i=1}^n
(N(T_i),N(T_i^*))=(N(F),N(F^*))=N(F,F^*)$ and that for every
$i\in\{1,\cdots,n\}$, $H_0(F)=H_0(T_i)$.
\end{proof}

%

\subsection{Replacing one edge of an OCRSF}\label{sec:replacingedges}

Let $F$ be an OCRSF of $G_1$, and consider an oriented edge 
$e=(x,y)$ of $E(G_1)\setminus F$. Then, since $F$ is an OCRSF, the vertex $x$
has a unique outgoing edge $e_x$ of $F$. 
Define $F_{\{e,e_x\}}$ to be the edge configuration $F$ where the edge $e_x$ is replaced by the edge $e$:
\begin{equation*}
F_{\{e,e_x\}}:=F\cup\{e\}\setminus\{e_x\}. 
\end{equation*}
In this section we describe features of the edge configuration $F_{\{e,e_x\}}$, which by construction 
contains one outgoing edge at every vertex.

Denote by $T_1,\cdots,T_n$ the CRT components of $F$, and
let $\gamma_i$  be the non-trivial cycle of the CRT component 
$T_i$. Without loss of generality let us suppose $n\geq 2$, the case $n=1$ simply being a boundary case of the case 
$n=2$. Then, the edge $e$ belongs to a unique cylinder $\Cscr$ obtained by cutting along two neighboring non trivial 
cycles of $F$, say $\gamma_1,\gamma_2$, see Figures \ref{fig:fig9}, \ref{fig:fig10}, \ref{fig:fig11}. Since $F$ is an OCRSF, 
the vertex $x$ (resp. $y$) is connected by an oriented edge-path $p_x$ (resp.
$p_y$) of $F$ to one of the non-trivial cycles 
$\gamma_1$ or $\gamma_2$. Let $F^*$ be a dual OCRSF 
of $F$ (at this point the orientation of its non-trivial cycles does not matter), then there is a unique 
non-trivial cycle $\gamma^*$ of $F^*$ contained in the cylinder $\Cscr$. Cutting along $\gamma^*$ separates 
the cylinder $\Cscr$ into two disjoint cylinders $\Cscr_1$, $\Cscr_2$. In order to analyze the features 
of the edge configuration $F_{\{e,e_x\}}$, we first describe the edge configuration $F\cup\{e\}$, split 
according to the following cases.

\begin{enumerate}
\item[Case $1$:] The edge $e$ belongs to the cylinder $\Cscr_1$, see Figure \ref{fig:fig9} (the case where $e$ belongs to 
the cylinder $\Cscr_2$ is symmetric). Then, the vertices $x$ and $y$ are connected to the path 
$\gamma_1$, and the paths $p_x$ and $p_y$ are contained in the cylinder $\Cscr_1$. In Case $1$, 
we suppose moreover that the paths $p_x$, $p_y$ merge at a vertex $m$ before hitting $\gamma_1$ 
or exactly when hitting $\gamma_1$. Then, the edge configuration consisting of the edge $e$ and the 
part of $p_x$ and $p_y$ from $x$ (resp. $y$) to $m$ consists of:
\begin{enumerate}
\item[Case $1a$:] either a trivial cycle,
\item[Case $1b$:] or a non-trivial cycle parallel to $\gamma_1$.
\end{enumerate}

\begin{figure}[h]
\begin{center}
\includegraphics[height=3.7cm]{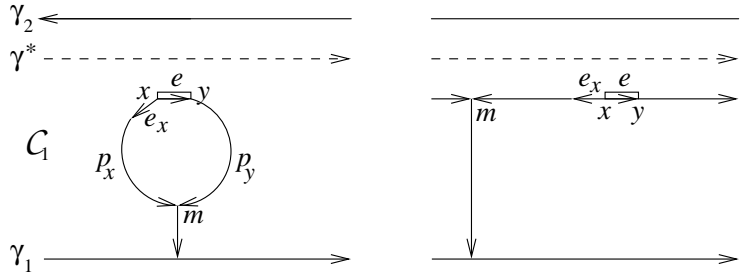}
\end{center}
\caption{The edge configuration $F\cup\{e\}$ in Case $1a$ (left) and Case $1b$ (right).}
\label{fig:fig9}
\end{figure}

\item[Case 2:] As in Case $1$, the edge $e$ belongs to the cylinder $\Cscr_1$, but this time the paths 
$p_x, p_y$ do not merge before hitting $\gamma_1$, see Figure \ref{fig:fig10}. Then cutting along 
$\{e\}\cup \{p_x\}\cup \{p_y\}$ separates the cylinder $\Cscr_1$ into two connected components, one homeomorphic to a disc and the other 
to a cylinder. We let $m$ be the unique boundary vertex of the component 
homeomorphic to a disc, with two incoming edges.\\
\vspace{2cm}

\begin{figure}[ht]
\begin{center}
\includegraphics[height=3.7cm]{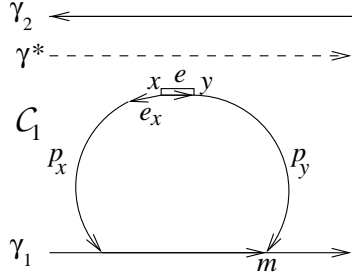}
\end{center}
\caption{The edge configuration $F\cup\{e\}$ in Case $2$.}
\label{fig:fig10}
\end{figure}

\item[Case 3:] The edge $e$ crosses the non-trivial cycle $\gamma^*$, see Figure \ref{fig:fig11}. 
Then, up to relabeling, the vertex $x$ (resp. $y$) is connected to the path $\gamma_1$ (resp. $\gamma_2$), 
and the path $p_x$ (resp. $p_y$) is contained in the cylinder $\Cscr_1$ (resp. $\Cscr_2$). 
We let $m$ be the vertex at which the path $\gamma_x$ hits the non-trivial cycle $\gamma_1$.\\

\begin{figure}[ht]
\begin{center}
\includegraphics[height=3.2cm]{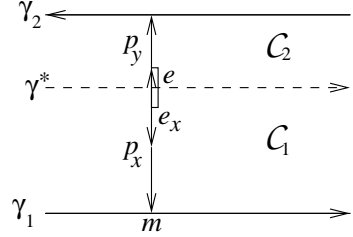}
\end{center}
\caption{The edge configuration $F\cup\{e\}$ in Case $3$.}
\label{fig:fig11}
\end{figure}
\end{enumerate}
As a consequence, we obtain the following characterization of the edge
configuration $F_{\{e,e_x\}}$. 

\begin{lem}\label{lem:replacing1edge}
Let $F$ be an OCRSF of $G_1$, $e=(x,y)$ be an edge of $E(G_1)\setminus F$, and
$e_x$ be the unique edge of $F$ exiting the vertex $x$. 
Then using the above case splitting, we have:
\begin{enumerate}
\item In all cases, as long as $x\neq m$, then the edge configuration $F_{\{e,e_x\}}$ is an OCRSF whose oriented non-trivial 
cycles are in bijection with those of $F$, so that $F$
and $F_{\{e,e_x\}}$ have the same reference number and $N(F_{\{e,e_x\}})=N(F)$.
\item When $x=m$, \emph{i.e.} the path $p_x$ is reduced to the point $m$, refer to Figure \ref{fig:fig13}:
\begin{itemize}
\item In Cases $1a$ and $2$, a trivial cycle is created so that $F_{\{e,e_x\}}$
is not an OCRSF.
\item In Case $1b$: when $x=m\notin\gamma_1$, then a non-trivial cycle parallel to $\gamma_1$ is created, so that 
$F_{\{e,e_x\}}$ is an OCRSF with the same reference number as $F$, satisfying
$N(F_{\{e,e_x\}})=N(F)\pm 1$. When $x=m\in\gamma_1$, then a non-trivial cycle
parallel to $\gamma_1$ is created 
and $\gamma_1$ is broken up. Thus,
$F_{\{e,e_x\}}$ is an OCRSF with the same reference number as $F$, satisfying 
$N(F_{\{e,e_x\}})=N(F)\pm 2$ or $N(F_{\{e,e_x\}})=N(F)$.
\item In Case $3$: as long as the number $n$ of non-trivial cycles is $\geq 2$,
then the non-trivial cycle $\gamma_1$ is broken up,
so that $F_{\{e,e_x\}}$ is an OCRSF with the same reference number as $F$,
satisfying
$N(F_{\{e,e_x\}})=N(F)\pm 1$. When $n=1$, then the unique non-trivial cycle 
is broken up, and a non-trivial cycle orthogonal to $\gamma_1$ is created. Thus
$F_{\{e,e_x\}}$ has a different reference number than $F$.
\end{itemize}
\begin{figure}[ht]
\begin{center}
\includegraphics[width=13cm]{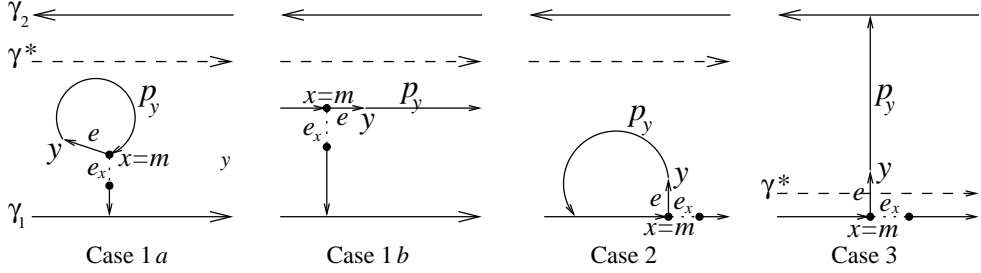}
\end{center}
\caption{The edge configuration $F_{\{e,e_x\}}$, when $x=m$.}
\label{fig:fig13}
\end{figure}
\end{enumerate}
\end{lem}

\subsection{Licit primal/dual edge moves}\label{sec:replacing2edges}

Let $(F,F^*)$ be a pair of dual OCRSFs of $G_1$ and $G_1^*$. We consider
an edge $e_1=(x_1,y_1)$ of $E(G_1)\setminus F$, and an edge
$e_2^*=(x_2^*,y_2^*)$ of 
$E(G_1^*)\setminus F^*$. Then, $e_1$ is the dual of an oriented edge
$e_1^*=(x_1^*,y_1^*)$ of $F^*$, and 
$e_2^*$ is the dual of an oriented edge $e_2=(x_2,y_2)$ of $F$. 
\begin{defi}\label{PrimalDualMove}
We say that the pair of primal and dual oriented edge configurations 
$\bigl(F_{\{e_1,e_2\}},F^*_{\{e_2^*,{e_1^*}\}}\bigr)$ 
is obtained from $(F,F^*)$ by a \emph{primal/dual edge move}. This move 
is called \emph{licit}, whenever:
\begin{equation}\label{Condition}
x_1=x_2\quad\text{ and}\quad x_1^*=x_2^*.
\end{equation}
\end{defi}

\begin{prop}\label{prop:replacing2edges}
The pair $(F_{\{e_1,e_2\}},F^*_{\{e_2^*,{e_1^*}\}})$ consists of dual OCRSFs of
$G_1$ and $G_1^*$, iff 
the primal/dual edge move is licit. When this is the case, 
\begin{enumerate}
 \item either the pair $(F_{\{e_1,e_2\}},F^*_{\{e_2^*,{e_1^*}\}})$ has the same
reference number as $(F,F^*)$, and
\begin{equation*}
 N(F_{\{e_1,e_2\}},F^*_{\{e_2^*,{e_1^*}\}})=N(F,F^*),
\end{equation*}
\item or the pairs $(F_{\{e_1,e_2\}},F^*_{\{e_2^*,{e_1^*}\}})$ and $(F,F^*)$
have different reference numbers, and
$$
N(F_{\{e_1,e_2\}},F^*_{\{e_2^*,{e_1^*}\}})=N(F,F^*)=0.
$$
\end{enumerate}
\end{prop}

\begin{proof}
Let us first show that Condition \eqref{Condition} is necessary. By construction of $F_{\{e_1,e_2\}}$, 
the outgoing edge $e_1$ is added at the vertex $x_1$. Since $F$ is an OCRSF, the
vertex $x_1$ has a unique outgoing 
edge $e_{x_1}$ of $F$. Suppose that $x_1\neq x_2$, this implies that $e_2\neq
e_{x_1}$ so that the vertex $x_1$ has two outgoing 
edges in the edge configuration $F_{\{e_1,e_2\}}$. By Remark~\ref{rem:OCRSF},
this forbids $F_{\{e_1,e_2\}}$ from being an OCRSF. 
A similar argument holds for $F^*_{\{e_2^*,e_1^*\}}$, when 
$x_1^*\neq x_2^*$. 

Let us now suppose that $x_1=x_2$ and $x_1^*=x_2^*$, see Figure \ref{fig:fig14}.
\begin{figure}[ht]
\begin{center}
\includegraphics[height=2.7cm]{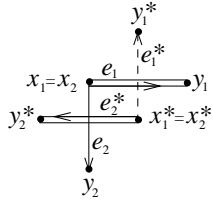}
\end{center}
\caption{The edges $e_1,e_2,e_1^*,e_2^*$, when $x_1=x_2$, $x_1^*=x_2^*$.}
\label{fig:fig14}
\end{figure}

In the sequel, we use the notations of Section \ref{sec:replacingedges}. Suppose
that $x_1\neq m_1$, and $x_2^*\neq m_2^*$. 
Then, by Point $1.$ of Lemma \ref{lem:replacing1edge}, the edge configuration
$F_{\{e_1,e_2\}}$ (resp. $F^*_{\{e_2^*,e_1^*\}}$ ) has the same reference number
as $F$
(resp. $F^*$) and $N(F_{\{e_1,e_2\}})=N(F)$ (resp.
$N(F^*_{\{e_2^*,e_1^*\}})=N(F^*)$), so that Point $1.$ of Proposition
\ref{prop:replacing2edges} is clearly satisfied.

Suppose now that $x_1=m_1$, then in Cases $1a$ and $2$, the edge-path $p_{y_1}$
hits the vertex $x_1$, and 
$\{e_1\}\cup\{p_{y_1}\}$ 
contains a trivial cycle, see Figure \ref{fig:fig15}. In the dual graph, the
edge-path $p_{x_1^*}$ must enter this trivial 
cycle, which is impossible since $p_{x_1^*}$ must also connect $x_1^*$ to one
of the dual non-trivial cycles, so that this case cannot occur.

\begin{figure}[ht]
\begin{center}
\includegraphics[height=3cm]{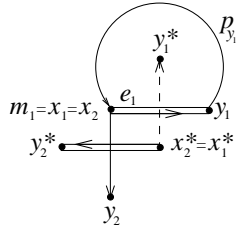}
\end{center}
\caption{In Cases $1a$ and $2$, when $x_1=m_1$, the edge-path $p_{y_1}$ hits
$x_1$, and $\{e_1\}\cup\{p_{y_1}\}$ contains a trivial cycle.}
\label{fig:fig15}
\end{figure}

In Case $1b$, the edge-path $p_{y_1}$ hits the vertex $x_1$, and
$\{e_1\}\cup\{p_{y_1}\}$ contains a non-trivial cycle parallel to 
$\gamma_1$, see Figure \ref{fig:fig16} (left). Suppose now that $x_1=m_1\notin
\gamma_1$, then in the dual graph, the edge path 
$p_{y_2^*}$ must hit $x_2^*$ (implying that $x_2^*=m_2^*$), and
$\{e_2^*\}\cup\{p_{y_2^*}\}$ must contain a non-trivial cycle parallel to
$\gamma_1$, implying that the non-trivial cycles contained in
$\{e_1\}\cup\{p_{y_1}\}$ and $\{e_2^*\}\cup\{p_{y_2^*}\}$ are parallel and in
opposite directions, see Figure \ref{fig:fig16} (right). Thus, by
Point $2.$, Case $1b$ of Lemma \ref{lem:replacing1edge}, we know that
$(F_{\{e_1,e_2\}},F^*_{\{e_2^*,e_1^*\}})$ are dual OCRSFs of $G_1$ and
$G_1^*$, with the same reference number as $(F,F^*)$, such that:
\begin{equation*}
N(F_{\{e_1,e_2\}},F^*_{\{e_2^*,e_1^*\}})=N(F)\pm 1+N(F^*)\mp 1=N(F,F^*). 
\end{equation*}

\begin{figure}[ht]
\begin{center}
\includegraphics[height=3.4cm]{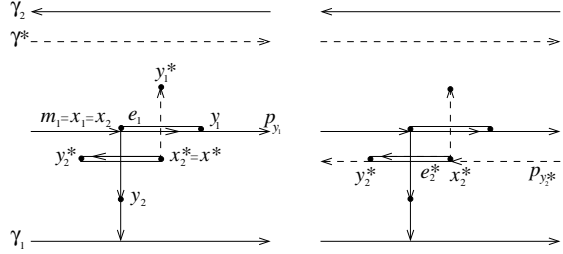}
\end{center}
\caption{Left: In Cases $1b$, when $x_1=m_1$, the edge-path $p_{y_1}$ hits
$x_1$, and $\{e_1\}\cup\{p_{y_1}\}$ contains a non-trivial cycle.
Right: when $x_1=m_1\notin\gamma_1$, the edge path $p_{y_2^*}$ must hit
$x_2^*$, and $\{e_2^*\}\cup\{p_{y_2^*}\}$ contains a non-trivial cycle, with
opposite direction.}
\label{fig:fig16}
\end{figure}

In Case $1b$, when $x_1=m_1\in\gamma_1$, then $\{e_1\}\cup\{p_{y_1}\}$ and
$\gamma_1$ must be co-oriented, so that by Point $2.$, Case $1b$ of 
Lemma \ref{lem:replacing1edge}, we know that $F_{\{e_1,e_2\}}$ is an OCRSF, in
which the non-trivial cycle $\gamma_1$ is broken up 
and the non-trivial cycle $\{e_1\}\cup\{p_{y_1}\}$ is added, implying that
$F_{\{e_1,e_2\}}$ has the same reference number as $F$, and
$N(F_{\{e_1,e_2\}})=N(F)$. 
In the dual graph, we have 
$x_2^*\neq m_2^*$, so that by Point $1.$ of Lemma \ref{lem:replacing1edge}, we
know that $F^*_{\{e_2^*,e_1^*\}}$ is an OCRSF whose non-trivial cycles are in
bijection with those of $F$. Summing the two contributions $N(F_{\{e_1,e_2\}})$
and $N(F^*_{\{e_2^*,e_1^*\}})$, we deduce 
that Point $1.$ of Proposition \ref{prop:replacing2edges} is satisfied.    

In Case $3$, when $x_1=m_1$, the edge $e_1^*$ must belong to $\gamma^*$
(implying that $x_1^*=x_2^*=m_2^*$), and 
the paths $\gamma^*$ and $\gamma_1$ must be in opposite directions, see Figure
\ref{fig:fig17}. If the number of non-trivial cycles of $F$ is $\geq 2$, then by
Point $2.$, Case $3$ of Lemma \ref{lem:replacing1edge}, we know that
$F_{\{e_1,e_2\}}$ (resp. $F^*_{\{e_2^*,e_1^*\}}$) is an OCRSF, in which the
non-trivial cycle $\gamma_1$ (resp. $\gamma^*$) is broken up. As a consequence, 
$(F_{\{e_1,e_2\}},F^*_{\{e_2^*,e_1^*\}})$ are dual OCRSFs of $G_1$ and $G_1^*$
with the same reference number as $(F,F^*)$ such that:
\begin{equation*}
N(F_{\{e_1,e_2\}},F^*_{\{e_2^*,e_1^*\}})=N(F)\pm 1+N(F^*)\mp 1=N(F,F^*). 
\end{equation*}

\begin{figure}[ht]
\begin{center}
\includegraphics[height=3cm]{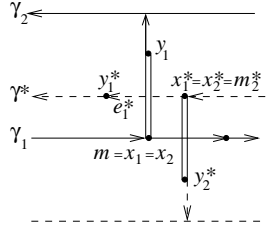}
\end{center}
\caption{In Case $3$, when $x_1=m_1$, the edge $e_1^*$ must belong to
$\gamma^*$, and the paths $\gamma^*$ and $\gamma_1$ must be in opposite
directions}
\label{fig:fig17}
\end{figure}

When the number of non-trivial cycles of $F$ is $1$, since $\gamma^*$ and
$\gamma_1$ must be in opposite directions, we know that $N(F,F^*)=0$. Moreover,
by Point $2.$, Case $3$ of Lemma \ref{lem:replacing1edge}, we know that
$F_{\{e_1,e_2\}}$ (resp. $F^*_{\{e_2^*,e_1^*\}}$) is an OCRSF, in which the
non-trivial cycle $\gamma_1$ (resp. $\gamma^*$) is broken up and in which a
non-trivial cycle orthogonal to the original one is
created. Since these must have opposite directions, we deduce that the pair
$(F_{\{e_1,e_2\}},F^*_{\{e_2^*,e_1^*\}})$ has a different reference number than
$(F,F^*)$,
and
$N(F_{\{e_1,e_2\}},F^*_{\{e_2^*,e_1^*\}})=0=N(F,F^*)$, implying that Point $2.$
of Proposition \ref{prop:replacing2edges} is satisfied.
\end{proof}

\section{Proof of Theorem \ref{thm:main2}}\label{sec:CORRES}

Let $G$ be an infinite, $\ZZ^2$-periodic isoradial graph, $\GD$ be the
corresponding Fisher graph, and $G_1$, $\GD_1$ be their respective fundamental
domains. As a consequence of Lemma \ref{lem:rewriting} and using notations
introduced in Section \ref{sec:notations}, we deduce that the critical
Laplacian characteristic polynomial $\Plap(\zs,\ws)$, and the critical dimer
characteristic polynomial $\Pdimer^0(\zs,\ws)$, can be rewritten as:
\begin{align*}
\Plap(\zs,\ws)&=\sum_{(\Fb,\Fb^*)\in\F(G_1,G_1^*)}
\left(\prod_{\eb=\xb\yb\in
\Fb}\tan\theta_{\xb\yb}\right)(-\zs^{h_0}\ws^{v_0})^{\frac{1}{2}N(\Fb,\Fb^*)},\\
\Pdimer^0(\zs,\ws)&=\sum_{(F,F^*)\in\F^0(\GD_1,\GD_1^*)}
\left(\prod_{e=(x,y)\in F}f_x\overline{f_y}K_{x,y}
\right)
(-\zs^{h_0}\ws^{v_0})^{\frac{1}{2}N(F,F^*)}.
\end{align*}
In Section \ref{sec:explicitcorrespondence}, we prove the first part of Theorem
\ref{thm:main2} by constructing pairs of essential OCRSFs of $\GD_1$ and
$\GD_1^*$ counted by $\Pdimer^0(\zs,\ws)$ from
pairs of dual
OCRSFs of $G_1$ and $G_1^*$ counted by $\Plap(\zs,\ws)$, see Theorem
\ref{thm:correspondence}.
Then, in Section~\ref{sec:weightpreserving}, we prove that this construction
is weight
preserving, thus ending the proof of Theorem~\ref{thm:main2}. 

\subsection{Explicit construction}\label{sec:explicitcorrespondence}

Let us start by giving the general idea of the construction.
Let $(\Fb,\Fb^*)$ be a pair of dual OCRSFs
of $G_1$ and $G_1^*$, then to $(\Fb,\Fb^*)$ we assign a family
$\mathcal{S}(\Fb,\Fb^*)$ consisting of pairs of dual essential OCRSFs of
$\GD_1$ and $\GD_1^*$, such that:
\begin{enumerate}
\item $\bigcup_{(\Fb,\Fb^*)\in
\F(G_1,G_1^*)}\mathcal{S}(\Fb,\Fb^*)=\F^0(\GD_1,\GD_1^*).$
\item When $(\Fb_1,\Fb_1^*)\neq (\Fb_2,\Fb_2^*)$, then
$\mathcal{S}(\Fb_1,\Fb_1^*)\cap \mathcal{S}(\Fb_2,\Fb_2^*)=\emptyset$.
\item For every $(F,F^*)\in \mathcal{S}(\Fb,\Fb^*)$, then
\begin{itemize}
 \item either the reference number of $(F,F^*)$ in $\GD_1,\,\GD_1^*$ is equal to
the reference number
of $(\Fb,\Fb^*)$ in $G_1,\,G_1^*$, and
$N_{\GD_1,\GD_1^*}(F,F^*)=N_{G_1,G_1^*}(\Fb,\Fb^*)$,
\item or the pairs $(F,F^*)$ and $(\Fb,\Fb^*)$ have different reference numbers,
and\\
$N_{\GD_1,\GD_1^*}(F,F^*)=N_{G_1,G_1^*}(\Fb,\Fb^*)=0$,
\end{itemize}
where we have added a subscript to the signed number of cycles $N(\cdot)$ to
indicate on which graph it is
computed.
\end{enumerate}

For every pair $(\Fb,\Fb^*)$ of dual OCRSFs of $G_1$ and
$G_1^*$, loosely stated, the family $\mathcal{S}(\Fb,\Fb^*)$ is constructed as
follows. Let
$\eb_1,\cdots,\eb_m,$ be an arbitrary labeling of 
unoriented edges of $E(G_1)\setminus\Fb$. 
For $k\in\{0,\cdots,m\}$, we let 
$J_k=\{(i_1,\cdots,i_k)\in\{1,\cdots,m\}^k\,|\,1\leq i_1<\cdots<i_k\leq m\}$,
with the convention that
$J_k=\emptyset$, when $k=0$. Then,
$$
\S(\Fb,\Fb^*)=\bigcup_{k=0}^m \bigcup_{(i_1,\cdots,i_k)\in
J_k}\F^{(\Fb,\Fb^*),\{\eb_{i_1},\cdots,\eb_{i_k}\}}(\GD_1,\GD_1^*),
$$ 
where $\F^{(\Fb,\Fb^*),\{\eb_{i_1},\cdots,\eb_{i_k}\}}(\GD_1,\GD_1^*)$ is
constructed by induction on $k$ using licit primal/dual moves, introduced in
Section \ref{sec:general}. In Section
\ref{sec:initial}, we prove all results needed for the initial step of
the induction. Then, in Section \ref{sec:induction_step}, we specify licit
primal/dual moves to pairs of essential OCRSFs of $\GD_1$ and $\GD_1^*$. This
allows us to precisely define the set
$\F^{(\Fb,\Fb^*),\{\eb_{i_1},\cdots,\eb_{i_k}\}}(\GD_1,\GD_1^*)$ in
Section \ref{sec:conclusion}. Finally, we state and prove Theorem
\ref{thm:correspondence} establishing that we exactly obtain all pairs of dual
essential CRSFs of $\GD_1$ and $\GD_1^*$.

\subsubsection{Initial step of the induction} \label{sec:initial}

Let $\Fb$ be a subset of oriented edges of $G_1$ such that each vertex has
exactly one outgoing edge of $\Fb$.
Considering $\Fb$ as a subset of long edges of $\GD_1$ defines, for
every decoration $\xb$, a unique root vertex $v_1(\xb)$.
From now on, whenever no confusion occurs, we omit the argument $\xb$. For the
whole of this section, we let $\tau\in\{cw,cclw\}^{V(G_1)}$ be an assignment
of type \emph{cw} or \emph{cclw} to vertices of $G_1$.

\begin{lem}\label{lem:lem12}
A subset of oriented edges $F$ of $\GD_1$ is an essential OCRSF compatible
with $\Fb$ and $\tau$, iff:
\begin{enumerate}
\item $\Fb$ is an OCRSF of $G_1$,
\item long edges of $F$ are exactly those of $\Fb$,
\item the restriction of $F$ to every decoration $\xb$ contains:
\begin{itemize}
\item all inner edges with the orientation induced by the type $\tau(\xb)$,
\item the edge $(w_1,v_1)$ at the triangle $t_1$ if the
decoration is of type \emph{cw},
or the edge $(z_1,v_1)$ if it is of type $\emph{cclw}$,
\item for every $i\neq 1$, any of the three possible $2$-edge configuration of
Definition~\ref{def:restricted}, at the
triangle $t_i$, with the orientation induced by the type, see also Figure
\ref{fig:fig19}.
\end{itemize}
\end{enumerate}
Moreover, when this is the case, oriented non-trivial cycles of $F$ in $\GD_1$
are in bijection with oriented non-trivial cycles of $\Fb$ in $G_1$.
\end{lem}
\begin{proof}
Since each vertex of $G_1$ has a unique outgoing edge of $\Fb$, we know that
$\Fb$ is an OCRSF of $G_1$, iff it contains
no trivial cycle. Moreover, we know that $F$ is an essential OCRSF of $\GD_1$ compatible with $\Fb$ and $\tau$, iff it 
satisfies Conditions $1.2.3.$ of Corollary \ref{rem:4}. Since every decoration of $\GD_1$ has a unique
root vertex, Condition $3.$ can be rewritten as Condition $3.$ of Lemma
\ref{lem:lem12} above. Thus, it remains to show that $F$ contains
no trivial cycle consisting of long and short edges of $\GD_1$, iff $\Fb$
contains no trivial cycle of $G_1$, or equivalently
$F$ contains a trivial cycle consisting of long and short edges iff $\Fb$
contains a trivial cycle. Because of the geometry of $\GD_1$, 
the direction from left to right
is straightforward. Conversely, suppose that $\Fb$ contains a trivial cycle.
Then, since every vertex of $G_1$ has a unique outgoing edge of 
$\Fb$ this implies that the trivial cycle must be co-oriented. Now, consider an
oriented configuration of $\GD_1$
satisfying Conditions $2.$ and $3$. above. Since every decoration of $\GD_1$ has
a unique root vertex, this means that every vertex of the 
decoration $\xb$ is connected by a unique path to the root vertex $v_1(\xb)$, in particular this holds for every vertex of type 
`$v$' of the decoration. Thus, if $\Fb$ contains a trivial cycles, then $F$ contains a trivial cycle consisting of 
long and short edges. A similar argument shows that if $F$ is an OCRSF of
$\GD_1$, its non-trivial cycles are in bijection with 
those of $\Fb$ in $G_1$.
\end{proof}

Figure \ref{fig:fig21} below illustrates in an example all essential OCRSFs of
$\GD_1$ compatible with $\Fb$ and $\tau$.\\

\begin{figure}[ht]
\begin{center}
\includegraphics[width=9.3cm]{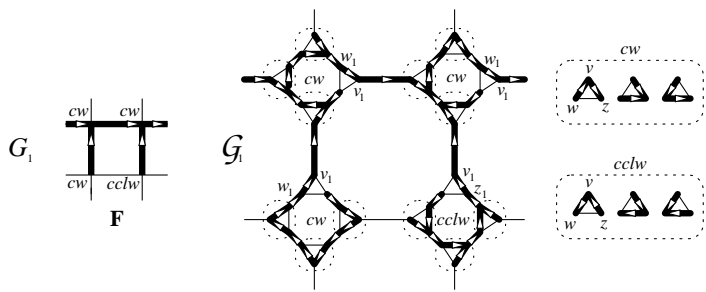}
\end{center}
\caption{Left: an OCRSF $\Fb$ of $G_1$ and an assignment of types to vertices
of $G_1$. Right: essential OCRSFs of $\GD_1$ compatible with $\Fb$ and $\tau$.}
\label{fig:fig21}
\end{figure} 

We now study properties of dual OCRSFs. Observe that
the dual graph $G_1^*$ is a subgraph of 
$\GD_1^*$, see also Figure \ref{fig:fig23} below. Since inner edges of the
decorations are always present in
essential OCRSFs of $\GD_1$, we omit their dual edges in our picture
of~$\GD_1^*$.

\begin{figure}[h]
\begin{center}
\includegraphics[width=6cm]{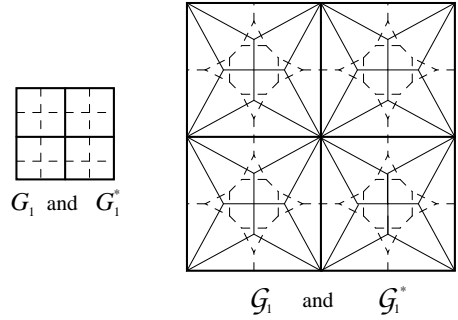}
\end{center}
\caption{Left: the graph $G_1$ (dotted lines) and it dual $G_1^*$ (full
lines). Right: the graph $\GD_1$ (dotted lines)
and its dual $\GD_1^*$ (full lines), without dual
edges of inner edges of the decorations.}
\label{fig:fig23}
\end{figure} 

Let $\Fb$ be an OCRSF of $G_1$ and $F$ be an essential OCRSF of 
$\GD_1$ compatible with $\Fb$ and $\tau$. Suppose for the moment that edges are unoriented, and 
denote by $F^*$ the dual CRSF of $F$ in 
$\GD_1^*$, and by $\Fb^*$ the dual CRSF of $\Fb$ in $G_1^*$. 
Then, the restriction of $F^*$ to dual long edges is $\Fb^*$, and since $G_1^*$
is a subgraph of $\GD_1^*$, the CRSF $\Fb^*$
is a subgraph of~$F^*$.

Moreover, by Lemma \ref{lem:lem12}, we know that
non-trivial cycles of $F$ in $\GD_1$ and $\Fb$ in $G_1$ are in bijection,
implying that 
non-trivial cycles of $F^*$ in $\GD_1^*$ and $\Fb^*$ in $G_1^*$ are also in
bijection.
As a consequence, the non-trivial cycles of $F^*$ in $\GD_1^*$ are exactly the non trivial cycles 
of $\Fb^*$ in $G_1^*$, and branches of $\Fb^*$ in $G_1^*$ are also branches of
$F^*$ in
$\GD_1^*$. Recalling that OCRSFs are obtained from a CRSF by orienting branches towards the non-trivial
cycles, and orienting each of the non-trivial cycles in one of the two possible
ways, we have shown
the following lemma describing oriented versions of $F^*$ and $\Fb^*$.

\begin{lem}\label{lem:lem13}
Let $\Fb$ be an OCRSF of $G_1$ and $F$ be an essential OCRSF of 
$\GD_1$ compatible with $\Fb$ and $\tau$. Then, if $F^*$ is a dual OCRSF of $F$, the restriction
$\Fb^*$ of $F^*$ to dual long edges is an OCRSF of $G_1^*$, and
oriented non-trivial cycles of $F^*$ in $\GD_1^*$ are exactly those of
$\Fb^*$
in $G_1^*$. Conversely, let $\Fb^*$ be a dual OCRSF of $\Fb$, then there is a
unique dual 
OCRSF $F^*$ of $F$ in $\GD_1^*$ such that the restriction of $F^*$ to dual long edges is $\Fb^*$. 
\end{lem}

Figure \ref{fig:fig22} below illustrates a dual OCRSF $\Fb^*$ of Figure
\ref{fig:fig21}, and the unique dual OCRSF 
of $F$ of Figure \ref{fig:fig21}, whose restriction to dual long edges is
$\Fb^*$.\\
\vspace{1cm}

\begin{figure}[h]
\begin{center}
\includegraphics[width=6cm]{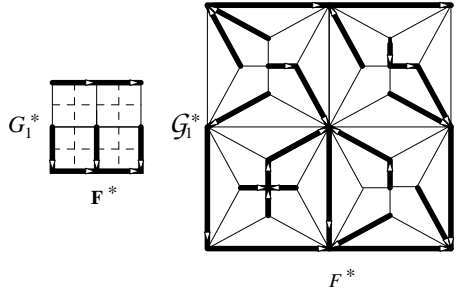}
\end{center}
\caption{Left: A dual OCRSF $\Fb^*$ of $\Fb$ of Figure \ref{fig:fig21}. Right:
the unique dual OCRSF $F^*$ of $F$ of Figure \ref{fig:fig21}
whose restriction to dual long edges is $\Fb^*$.}
\label{fig:fig22}
\end{figure}

\begin{defi}
Let $(\Fb,\Fb^*)$ be a pair of dual OCRSFs of $G_1$ and
$G_1^*$. Then, 
define $\F^{\tau,(\Fb,\Fb^*)}(\GD_1,\GD_1^*)$ to be the set of pairs $(F,F^*)$
of \emph{essential dual OCRSFs of $\GD_1$ and~$\GD_1^*$, compatible with
$(\Fb,\Fb^*)$ and $\tau$}, that is the set of pairs $(F,F^*)$,
such that $F$ is an essential OCRSF of $\GD_1$ compatible with $\Fb$ and $\tau$,
and $F^*$ is the unique dual OCRSF of $F$ whose restriction to dual long edges
is $\Fb^*$.
\end{defi}

As a consequence of Lemmas \ref{lem:lem12} and \ref{lem:lem13}, we have the
following.

\begin{cor}\label{cor:initial}
The set $\F^{\tau,(\Fb,\Fb^*)}(\GD_1,\GD_1^*)$ is well defined and non-empty. 
Moreover, the reference number in $\GD_1$ of every pair $(F,F^*)$ in this set is
equal to
the reference number in $G_1$ of $(\Fb,\Fb^*)$ in $G_1$, and
\begin{equation}\label{equ:NFF}
N_{\GD_1,\GD_1^*}(F,F^*)=N_{G_1,G_1^*}(\Fb,\Fb^*).
\end{equation}
\end{cor}

\subsubsection{Essential moves and reverse moves}\label{sec:induction_step}

Let $\Lb$ be a subset of oriented edges of $G_1$ also
considered as a subset of long edges of $\GD_1$, satisfying $(\ast)$ of
Section \ref{sec:OCRSFdecograph} (that is every decoration of $\GD_1$ has at
least one root vertex of $\Lb$). For the whole of this section, we let
$\tau\in\{cw,cclw\}^{V(G_1)}$ be an assignment of types to vertices of $G_1$.
Suppose that there exists an essential OCRSF $F$ of $\GD_1$ compatible with
$\Lb$ and $\tau$, and let $F^*$ be a dual OCRSF of $F$.

In this section, we characterize licit primal/dual moves performed on
$(F,F^*)$, yielding a pair of \emph{essential} dual
OCRSFs of $\GD_1$ and $\GD_1^*$, such that the first component either
has an additional long edge (\emph{essential move}) or a long
edge of $\Lb$ removed (\emph{essential reverse move}).

\textbf{Notation}. For the remainder of the paper, we need to introduce a
specific notation for oriented edges of $G_1$. Fix an arbitrary orientation of
edges of $G_1$, and denote by $\eb$ a generic unoriented edge, by $+\eb$ the
oriented edge compatible with the fixed orientation, and by $-\eb$ the reverse
oriented edge.

\underline{\textbf{Essential moves}}

Assume that $\Lb\neq E(G_1)$, and let $\eb$ be an unoriented edge of
$E(G_1)\setminus \Lb$. We now characterize licit
primal/dual moves performed on $(F,F^*)$, which yield a pair of essential dual 
OCRSFs of $\GD_1$ and $\GD_1^*$, such that the first component is an essential
OCRSF compatible with $\Lb\cup\{\pm\eb\}$ and $\tau$, i.e. belongs to 
$\F^{\tau,\Lb\cup\{\pm\eb\}}(\GD_1)$.

Suppose first that we want to add the edge $+\eb$ to $F$, 
corresponding to a long edge $(v_k(\xb),v_\l(\yb))$ of 
$\GD_1$, for some $k\in\{1,\cdots,d_{\xb}\},\,\l\in\{1,\cdots,d_{\yb}\}$, see
Figure~\ref{fig:fig18}. Assume for the moment 
that $\xb$ has type $\tau(\xb)=cw$, and let us omit the arguments $\xb$ and
$\yb$, recalling that the index $k$ refers to 
$\xb$, and the index $\l$ to $\yb$. 

Since $(v_k,v_\l)$ is absent in $F$, its dual edge $(x^*,y^*)$ is present in~$F^*$. Let us first handle Case A, 
where $x^*$ is to the right of $(v_k,v_\l)$, see Figure \ref{fig:fig18} (first
two columns). Since $F$ is an essential 
OCRSF it has one of the three possible $2$-edge configurations at the triangle
of the vertex $v_k$, with the orientation 
induced by the type $\tau(\xb)=cw$. Since we want the resulting configuration to
be an essential OCRSF compatible 
with $\tau$, this implies that after the move, the triangle $t_k$ must contain
either the edge $(w_k,v_k)$ or 
the edge $(w_k,z_k)$. Then, see Figure \ref{fig:fig18} (first column), our only choice is to remove the edge $(v_k,z_k)$ in 
Cases AI and AII, and the edge $(v_k,w_k)$ in Case AIII. This constraint fully determines the dual edges which are removed/added, 
and we only allow the move when it is \emph{licit}, i.e. when edges involved satisfy Condition \eqref{Condition}, that is in 
Cases AI and AII, see Figure \ref{fig:fig18} (second column). In Case $B$, when
$x^*$ is to the left of $(v_k,v_\l)$, a similar argument holds and yields the
last two columns of Figure \ref{fig:fig18}.\\

\begin{figure}[ht]
\begin{center}
\includegraphics[width=13cm]{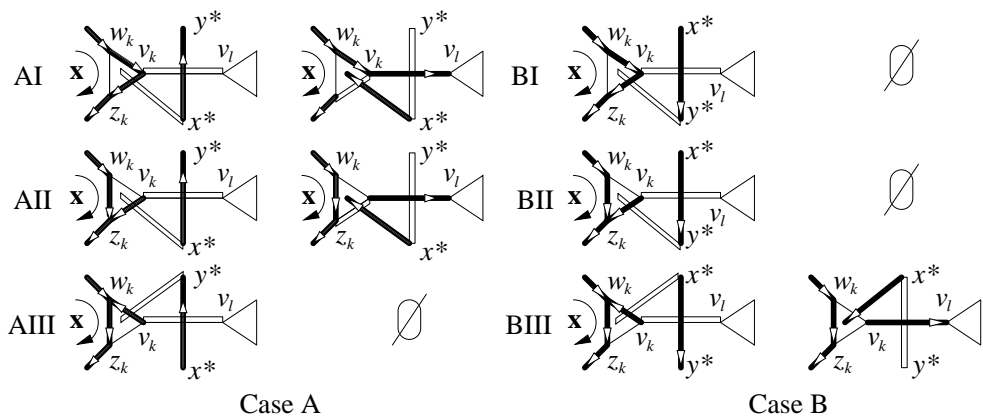}
\end{center}
\caption{Adding the long edge $(v_k,v_\l)$. First two columns: $x^*$ is to the right of $(v_k,v_\l)$. Last two columns: $x^*$ is to the left of $(v_k,v_\l)$.}
\label{fig:fig18}
\end{figure} 

By symmetry, when $\xb$ has type $\tau(\xb)=cclw$, if $x^*$ is to the right
(resp. to the left) of $(v_k,v_\l)$, we get
Case B (resp. Case A) with `$w$'s and `$z$'s exchanged. Using
symmetries again, the case where we add the edge $(v_\l,v_k)$ to the OCRSF $F$
is also
handled by Cases A and B. It is important to note that in each case at most one
primal/dual move is allowed.

\begin{defi}
When the above move is allowed, it is called an \emph{essential move}, and
$(F,F^*)_{\{\pm\eb,\cdot\}}=(F_{\{\pm\eb,\cdot\}},F_{\{\cdot,\pm\eb^*\}}^*)$
denotes the resulting pair of OCRSFs. When it is not allowed, we set by
convention, $(F,F^*)_{\{\pm\eb,\cdot\}}=\emptyset$.
\end{defi}

Then, Proposition \ref{prop:replacing2edges} and the construction of essential
moves immediately yields the following.

\begin{prop}\label{prop:essential}
Let $F$ be an essential OCRSF of $\GD_1$ compatible with $\Lb$ and $\tau$, and
$F^*$ be a dual OCRSF of $F$. Then, the pair $(F,F^*)_{\{\pm\eb,\cdot\}}$
consists of:
\begin{itemize}
\item the empty set when the essential move is not allowed,
\item when the essential move is allowed, a pair of dual essential OCRSFs of
$\GD_1$ and $\GD_1^*$, such that 
$F$ is an essential OCRSF compatible with $\Lb\cup\{\pm\eb\}$ and $\tau$, and 
\begin{enumerate}
 \item either the pair $(F,F^*)_{\{\pm\eb,\cdot\}}$ has the same reference
number as
$(F,F^*)$, and
$$
N((F,F^*)_{\{\pm\eb,\cdot\}})=N(F,F^*),
$$
\item or the pairs $(F,F^*)_{\{\pm\eb,\cdot\}}$ and $(F,F^*)$ have different
reference numbers, and
$$
N((F,F^*)_{\{\pm\eb,\cdot\}})=N(F,F^*)=0.
$$
\end{enumerate}
\end{itemize}
Moreover, the orientation of edges not involved in the move remains unchanged.
\end{prop}

\underline{\textbf{Essential reverse moves}}

Assume now that $\Lb\neq\emptyset$, and let $\epsilon\eb=(\xb,\yb)$ be an
oriented edge of $\Lb$, where $\epsilon\in\{-,+\}$.
We now characterize licit primal/dual moves performed on $(F,F^*)$, which yield
a pair of essential OCRSFs of $\GD_1$ and $\GD_1^*$, such that the first
component is an essential OCRSF compatible with $\Lb\setminus\{\epsilon\eb\}$
and $\tau$, i.e. belongs to $\F^{\tau,\Lb\setminus\{\epsilon\eb\}}(\GD_1)$.

The edge $\epsilon\eb$ corresponds to a long edge $(v_k(\xb),v_\l(\yb))$ of
$\GD_1$,
see Figure \ref{fig:fig18a}. Assume for the moment
that $\xb$ is of type $\tau(\xb)=cw$, and let us omit the arguments $\xb$ and
$\yb$ onwards. Since $F$ 
is an essential OCRSF, it contains either the edge $(w_k,v_k)$ or the edge
$(w_k,z_k)$ at the 
triangle $t_k$. Note that the edge $z_k v_k$ is always absent, so
that its dual 
edge is present in $F^*$. Let us first handle Case C, where this dual edge is oriented towards 
the triangle $t_k$.

Since we want the resulting configuration to be an essential OCRSF compatible
with $\tau$, this 
implies that after the move, the triangle $t_k$ must contain one of the three
possible 
$2$-edge configurations, with the orientation induced by $\tau(\xb)=cw$. Then,
see Figure 
\ref{fig:fig18a}, when the triangle $t_k$ contains the edge $(w_k,v_k)$, our
only choice 
is to add the edge $(v_k,z_k)$ (Case CI). When it contains the edge $(w_k,z_k)$, we can either 
add the edge $(v_k,z_k)$ (Case CII) or the edge $(v_k,w_k)$ (Case CIII). We only allow this move 
when it is licit, that is in Cases CI and CII.

\begin{figure}[ht]
\begin{center}
\includegraphics[width=13cm]{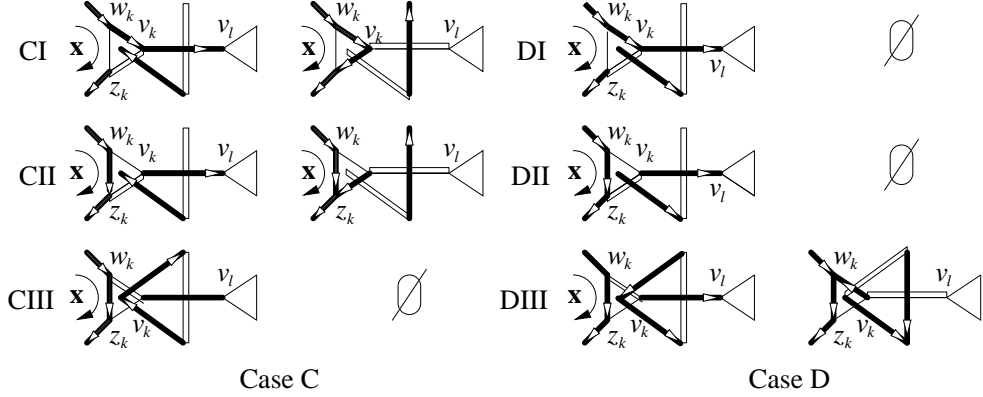}
\end{center}
\caption{Removing the long edge $(v_k,v_\l)$ when the dual of the edge
$z_k v_k$ is oriented towards (away from) the triangle $t_k$ in Case C
(Case D).}
\label{fig:fig18a}
\end{figure}

In Case D, where the dual of the edge $z_k v_k$ is oriented away from the
triangle $t_k$, a similar argument holds, and yields the last two columns of
Figure \ref{fig:fig18a}, the move is only allowed in Case DIII. It is important
to note that in each case, at most one move is allowed. By symmetry, if $\xb$
has type \emph{cclw}, Cases C and D hold, with `$w$'s and `$z$'s exchanged.

\begin{defi}
When the above move is allowed, we refer to it as an \emph{essential
reverse move}, and denote by $(F,F^*)_{\{\cdot,\epsilon\eb\}}=
(F_{\{\cdot,\epsilon\eb\}},F_{\{\epsilon\eb^*,\cdot\}})$ the
resulting
pair of OCRSFs. When the move is not 
allowed, we set by convention $(F,F^*)_{\{\cdot,\epsilon\eb\}}=\emptyset$.
\end{defi}

The analogous of Proposition \ref{prop:essential} also holds for essential
reverse moves, and we do not write it out explicitly. The next lemma relates
essential moves and reverse ones. 

\begin{lem}\label{lem:essentialreverse}
Let $(F,F^*)$ be a pair of dual essential OCRSFs of $\GD_1$ and $\GD_1^*$, such
that $F$ is compatible with $\Lb$ and $\tau$.
\begin{enumerate}
\item Let $\eb$ be an unoriented edge of $E(G_1)\setminus\Lb$, and $\epsilon\in
\{-,+\}$. Then, if an essential move can be performed on $(F,F)^*$, yielding a
pair $(F,F^*)_{\{\epsilon\eb,\cdot\}}$, the essential reverse 
primal/dual move, which removes the long edge $\epsilon\eb$ from
$(F,F^*)_{\{\epsilon\eb,\cdot\}}$, can also 
be performed and:
\begin{equation*}
((F,F^*)_{\{\epsilon\eb,\cdot\}})_{\{\cdot,\epsilon\eb\}}=(F,F^*). 
\end{equation*}
\item Let $\epsilon\eb$ be an oriented edge of $\Lb$, where
$\epsilon\in\{-,+\}$. Then, if an essential reverse primal/dual move can be
performed on $(F,F^*)$, yielding a pair $(F,F^*)_{\{\cdot,\epsilon\eb\}}$, the
essential 
primal/dual move, which adds the long edge $\epsilon\eb$ to
$(F,F^*)_{\{\cdot,\epsilon\eb\}}$, can also 
performed, and:
\begin{equation*}
 ((F,F^*)_{\{\cdot,\epsilon\eb\}})_{\{\epsilon\eb,\cdot\}}=(F,F^*).
\end{equation*}
\end{enumerate}
\end{lem}

\begin{proof}
From Figures \ref{fig:fig18} and \ref{fig:fig18a}, we immediately check that the moves 
AI/CI, AII/CII, BIII/DIII are inverse of eachother. 
\end{proof}

\subsubsection{Explicit construction}\label{sec:conclusion}

Let $(\Fb,\Fb^*)$ be a pair of dual OCRSFs of $G_1$ and $G_1^*$, and let
$\eb_1,\cdots,\eb_m,$ be an arbitrary labeling of 
unoriented edges of $E(G_1)\setminus\Fb$. We let 
$J_k=\{(i_1,\cdots,i_k)\in\{1,\cdots,m\}^k\,|\,1\leq i_1<\cdots<i_k\leq m\}$,
with the convention that $J_k=\emptyset$, when $k=0$. Let
$\tau\in\{cw,cclw\}^{V(G_1)}$ be an assignment of types to vertices of~$G_1$.

We now define by induction, for every $k\in\{0,\cdots,m\}$ and
every $(i_1,\cdots,i_k)\in~J_k$, the set, 
$$\F^{\tau,(\Fb,\Fb^*),\{\eb_{i_1},\cdots,\eb_{i_k};\cdot\}}(\GD_1,\GD_1^*).$$

\underline{Initial step}:
$
\F^{\tau,(\Fb,\Fb'),\emptyset}(\GD_1,\GD_1^*)=\F^{\tau,(\Fb,\Fb^*)}(\GD_1,\GD_1^*),
$
defined in Section \ref{sec:initial}.

\underline{Induction step}: for $j=1,\cdots,k$,
\begin{align*}
\F^{\tau,(\Fb,\Fb^*),\{\eb_{i_1},\cdots,\eb_{i_j};\cdot\}}(\GD_1,\GD_1^*)=\quad\quad\quad\quad\quad&\\
\bigcup_{(F,F^*)\in \F^{\tau,(\Fb,\Fb^*),\{\eb_{i_1},\cdots,\eb_{i_{j-1}};\cdot\}}(\GD_1,\GD_1^*)}&
\{(F,F^*)_{\{+\eb_{i_j},\cdot\}}\cup (F,F^*)_{\{-\eb_{i_j},\cdot\}}\},
\end{align*}
where $(F,F^*)_{\{\pm\eb_{i_j},\cdot\}}$ is defined in Section
\ref{sec:induction_step}. Then, we have:

\begin{prop}\label{lem:lem14}
The set
$\F^{\tau,(\Fb,\Fb^*),\{\eb_{i_1},\cdots,\eb_{i_k};\cdot\}}(\GD_1,\GD_1^*)$ is
non empty, independent
of the order in which edges are added, and consists of distinct pairs of
essential dual OCRSFs of $\GD_1$ and $\GD_1^*$, compatible with $\tau$.
Moreover, for every pair $(F,F^*)$ in this set,
\begin{enumerate}
 \item either the reference number of $(F,F^*)$ in $\GD_1$ is equal to the
reference number of
$(\Fb,\Fb)^*$ in $G_1$, and:
$$
N_{\GD_1,\GD_1^*}(F,F^*)=N_{G_1,G_1^*}(\Fb,\Fb^*),
$$
\item or the reference numbers differ, and 
$$
N_{\GD_1,\GD_1^*}(F,F^*)=N_{G_1,G_1^*}(\Fb,\Fb^*)=0.
$$
\end{enumerate}
\end{prop}

\begin{rem}\label{rem:correspondence}
Note that the set
$\F^{\tau,(\Fb,\Fb^*),\{\eb_{i_1},\cdots,\eb_{i_k};\cdot\}}(\GD_1,\GD_1^*)$ can
be rewritten as:
\begin{equation*}
\bigcup_{(\epsilon_{i_1},\cdots,\epsilon_{i_k})\in\{-,+\}^k}
\F^{\tau,(\Fb,\Fb^*),\{\epsilon_{i_1}\eb_{i_1},\cdots,\epsilon_{i_k}\eb_{i_k}
;\cdot\}}(\GD_1, \GD_1^*),
\end{equation*}
where
$\F^{\tau,(\Fb,\Fb^*),\{\epsilon_{i_1}\eb_{i_1},\cdots,\epsilon_{i_k}\eb_{i_k}
;\cdot\}}(\GD_1, \GD_1^*)$ is defined by induction as follows. The
initial set is $\F^{\tau,(\Fb,\Fb^*)}(\GD_1,\GD_1^*)$, and for every
$j\in\{1,\cdots,k\}$,
\begin{align*}
\F^{\tau,(\Fb,\Fb^*),\{\epsilon_{i_1}\eb_{i_1},\cdots,\epsilon_{i_j}\eb_{i_j}
;\cdot\}}&(\GD_1, \GD_1^*)\\
&=\bigcup_{(F,F^*)\in\F^{\tau,(\Fb,\Fb^*),\{\epsilon_{i_1}\eb_{i_1},\cdots,
\epsilon_{i_{j-1}}\eb_{i_{j-1}}
;\cdot\}}(\GD_1, \GD_1^*)} \{(F,F^*)_{\{\epsilon_{i_j}\eb_{i_j},\cdot\}}\}.
\end{align*}
\end{rem}

\begin{proof}
As a consequence of Remark \ref{rem:correspondence}, it suffices to show that
for every $(\epsilon_{i_1},\cdots,\epsilon_{i_k})\in\{-,+\}^k$, the set
$\F^{\tau,(\Fb,\Fb^*),\{\epsilon_{i_1}\eb_{i_1},\cdots,\epsilon_{i_k}\eb_{i_k}
;\cdot\}}(\GD_1, \GD_1^*)$ satisfies Proposition \ref{lem:lem14}. Let us first
prove that the initial set
$\F^{\tau,(\Fb,\Fb^*)}(\GD_1,\GD_1^*)$
satisfies all statements. By Lemma~\ref{lem:lem12} and~\ref{lem:lem13}, it is
clearly non empty, and consists of pairs of essential dual OCRSFs of $\GD_1$
and $\GD_1^*$, compatible with $\tau$; and by Corollary \ref{cor:initial},
it satisfies Point $1$.

Returning to the definition of essential moves and to the characterization of
OCRSFs of $\F^{\tau,(\Fb,\Fb^*)}(\GD_1,\GD_1^*)$, we deduce that 
the set
$\F^{\tau,(\Fb,\Fb^*),\{\epsilon_{i_1}\eb_{i_1},\cdots,\epsilon_{i_k}\eb_{i_k}
;\cdot\}}(\GD_1, \GD_1^*)$ is also non-empty. By
Proposition \ref{prop:essential}, we know that it consists of pairs of
essential dual OCRSFs of $\GD_1$ and $\GD_1^*$ compatible with~$\tau$.
Moreover, since the orientation of edges not
involved in the move remains unchanged, we deduce that the induction is in fact
independent of the order in which edges are added. Proposition
\ref{prop:essential} also implies that either Point $1.$ or $2.$ is satisfied.
\end{proof}

\begin{defi}
\begin{align*}
\S^{\tau,(\Fb,\Fb^*)}(\GD_1,\GD_1^*)&=\bigcup_{k=0}^m \bigcup_{
\{\{\eb_{i_1},\cdots,\eb_{i_k}\}\,:\,
(i_1,\cdots,i_k)\in J_k\}
} \F^{\tau,(\Fb,\Fb^*),\{\eb_{i_1},\cdots,\eb_{i_k};\cdot\}}(\GD_1,\GD_1^*)\\
\S^{(\Fb,\Fb^*)}(\GD_1,\GD_1^*)&=\bigcup_{\tau\in\{cw,ccwl\}^{V(G_1)}}
\S^{\tau,(\Fb,\Fb^*)}(\GD_1,\GD_1^*).
\end{align*}
\end{defi}

The next theorem proves that we exactly obtain all pairs of dual CRSFs of
$\GD_1$ and $\GD_1^*$.

\begin{thm}\label{thm:correspondence}
$\,$
\begin{enumerate}
\item $\displaystyle\bigcup_{\scriptstyle{
(\Fb,\Fb^*)\in\F(G_1,G_1^*)}}\S^{(\Fb,\Fb^*)}(\GD_1,\GD_1^*)
=\F^0(\GD_1,\GD_1^*)$.
\item When $(\Fb_1,\Fb_1^*)\neq (\Fb_2,\Fb_2^*)$, then $\S^{(\Fb_1,\Fb_1^*)}(\GD_1,\GD_1^*)
\cap\S^{(\Fb_2,\Fb_2^*)}(\GD_1,\GD_1^*)=
\emptyset$.
\end{enumerate}
\end{thm}

\begin{proof}
The inclusion $\subset$ of Point $1.$ is a direct consequence of Proposition
\ref{lem:lem14}. 

Consider a pair $(F,F^*)$ of dual essential OCRSFs of $\GD_1$ and
$\GD_1^*$, and let us show that there exists a pair of dual OCRSFs $(\Fb,\Fb^*)$
of $G_1$ and $G_1^*$, such that $(F,F^*)\in\S^{(\Fb,\Fb^*)}(\GD_1,\GD_1^*)$.
By definition of essential OCRSFs of $\GD_1$, there exists a subset
of edges $\Lb$ of $\boldsymbol{\mathcal{L}}$,
and $\tau\in\{cw,cclw\}^{V(G_1)}$, such that $F$ is an essential OCRSF
compatible with $\Lb$ and $\tau$, and $F^*$ is a dual
OCRSF of $F$. Recall that the set $\Lb$ defines, for every decoration $\xb$, a
set of root vertices $R_{\xb}(\Lb)$. 
Let us fix a decoration
$\xb$, assume that $\tau(\xb)=cw$ (the case where $\tau(\xb)=cclw$ being
similar), and omit the argument $\xb$ in the sequel.
By Corollary \ref{rem:4}, the restriction of $F$ to the decoration $\xb$
consists of:
\begin{enumerate}
\item[-] all inner edges oriented clockwise,
\item[-] one of the three possible $2$-edge configuration at the triangle of every non-root vertex, with the
appropriate orientation,
\item[-] one of the two following $1$-edge configurations at the triangle of
every root vertex 
$v_i\in R_{\xb}(\Lb)$: 
\begin{equation*}
\{(w_i,v_i)\},\{(w_i,z_i)\},
\end{equation*}
with the additional constraint that the triangle of at least one root vertex
contains the configuration $(w_i,v_i)$.
\end{enumerate}
Let us now study properties of the dual OCRSF $F^*$ at the decoration $\xb$, see
also Figure~\ref{fig:fig24}. Denote by $c^*$ the dual vertex at the center of
the decoration, and by
$t_1^*,\cdots,t_{d_{\xb}}^*$ the dual vertices
at the center of the triangles $t_1,\cdots,t_{d_{\xb}}$.
As inner edges are always present in the OCRSF $F$, we omit their dual
edges in our representation of the dual graph. 

\begin{figure}[ht]
\begin{center}
\includegraphics[width=4cm]{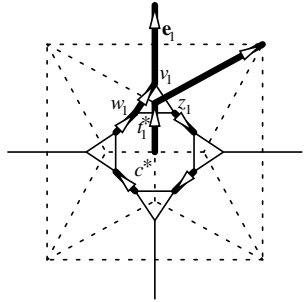}
\end{center}
\caption{The restriction of $F$ and $F^*$ to a decoration.}
\label{fig:fig24}
\end{figure}

Since $F^*$ is an OCRSF, there is exactly one edge
$(c^*,t_i^*)$ exiting the vertex $c^*$, we set by convention $i=1$.
Then, since there must also be one edge exiting
the vertex $t_1^*$, this means that $F$ must contain the edge $(w_1,v_1)$,
implying that $v_1$ is a root 
vertex of $\Lb$, and that the long edge $\epsilon_{1}\eb_1$, where
$\epsilon_1\in\{-,+\}$, whose initial vertex is $v_1$, belongs to~$\Lb$. 
Repeating this for every decoration of $\GD_1$ defines a subset of
oriented edges $\Fb$ of~$G_1$, 
\begin{equation}\label{equ:defF}
\Fb=\bigcup_{\xb\in
V(G_1)}\{\epsilon_1(\xb)\eb_1(\xb)\},
\end{equation}
such that every vertex of $G_1$ has exactly one outgoing edge of this subset.

Suppose that the decoration $\xb$ has another root vertex $v_i$ ($i\neq 1$),
and let $\epsilon_i\eb_i$ be the edge of $\Lb$ whose initial vertex is $v_i$.
Then, at the triangle $t_i$, the OCRSF $F$ consists of:
\begin{itemize}
\item either the edge $(w_i,z_i)$; then, the dual of the edges $w_i v_i$, $z_i
v_i$ belong to $F^*$. A priori, there are two
possible orientations for the dual edges, and which one it is is fixed by $F^*$.
In each of the two cases, either the essential reverse move CII or DIII can be
performed on the pair $(F,F^*)$, yielding a pair
$(F,F^*)_{\{\cdot,\epsilon_i\eb_i\}}$ of dual essential OCRSFs of $\GD_1$ and
$\GD_1^*$.
\item or the edge $(w_i,v_i)$; then, the dual edge $(t_i^*,c^*)$ belongs to
$F^*$ and must be oriented towards $c^*$. Indeed, otherwise the
vertex $c^*$ would have two exiting edges of $F^*$ which contradicts the fact of
being an OCRSF. 
This implies that the dual of the edge $z_i v_i$, which also belongs to $F^*$
is oriented towards $t_i^*$. As
a consequence, the essential reverse move CI can be performed on the
pair $(F,F^*)$, yielding a pair $(F,F^*)_{\{\cdot,\epsilon_i\eb_i\}}$ of dual 
essential OCRSFs of $\GD_1$ and $\GD_1^*$.
\end{itemize}
Since the orientation of edges not involved in the move remains unchanged by
essential reverse moves, we can repeat this for every root vertex different
from $v_1$ of the decoration $\xb$, and for every decoration $\xb$ of $\GD_1$.
By the analogous of Proposition \ref{prop:essential} for essential reverse
moves, this yields a pair $(\bar{F},\bar{F}^*)$ of dual essential OCRSFs of
$\GD_1$ and $\GD_1^*$, such that $\bar{F}$ is compatible with $\Fb$ and $\tau$.
Then, by Lemma \ref{lem:lem12}, this implies that $\Fb$ is an OCRSF of $G_1$.
Let $\Fb^*$
be the restriction to dual long edges of $\bar{F}^*$. Then, by
Lemma~\ref{lem:lem13}, $\Fb^*$ is a dual OCRSF of $\Fb$, and we deduce that:
$$
(\bar{F},\bar{F}^*)\in\F^{\tau,(\Fb,\Fb^*)}(\GD_1,\GD_1^*).
$$ 
By Lemma \ref{lem:essentialreverse}, if an essential reverse move can be
performed on a pair of dual OCRSFs, removing an oriented long edge, then the
essential move adding the same long edge can be performed to recover the
original pair. Applying this recursively, we deduce that:
$$
(F,F^*)\in\F^{\tau,(\Fb,\Fb^*),\{\Lb\setminus \Fb\}}(\GD_1,\GD_1^*)\subset
\S^{(\Fb,\Fb^*)}(\GD_1,\GD_1^*),
$$ 
thus proving Point $1.$ of Theorem \ref{thm:correspondence}. 

Let us now prove Point $2.$ Suppose that there are two distinct pairs
$(\Fb_1,\Fb_1^*)$ and $(\Fb_2,\Fb_2^*)$ of dual OCRSFs of $G_1$ and $G_1^*$,
such that:
$$
(F,F^*)\in \S^{(\Fb_1,\Fb_1^*)}(\GD_1,\GD_1^*)\cap
\S^{(\Fb_2,\Fb_2^*)}(\GD_1,\GD_1^*).
$$
Let us return to properties of $(F,F^*)$ at a fixed decoration $\xb$ of $\GD_1$.
In $F^*$, the dual edge of the edge $z_1 v_1$ is
oriented away from $t_1^*$. Now, suppose that the edge
$\epsilon_1\eb_1$ is added by an essential move. This implies that in
the dual graph, the dual of the edge $z_1v_1$ must be added.
Referring to Figure \ref{fig:fig18}, which describes possible essential moves,
we see that this dual edge is then oriented towards $t_1^*(\xb)$, which is a
contradiction. Thus, $\epsilon_1\eb_1$ must be an edge of the original OCRSF.
Since essential moves do not change the orientation of edges not
involved in the move, we repeat this argument for every decoration, and
deduce that:
$$
\Fb_1=\Fb_2=\Fb,
$$ 
where $\Fb$ is defined in Equation \eqref{equ:defF}. As a consequence
$(F,F^*)\in \F^{\tau,(\Fb,\Fb_1^*),\{\Lb\setminus\Fb\}}(\GD_1,\GD_1^*)\cap
\F^{\tau,(\Fb,\Fb_2^*),\{\Lb\setminus\Fb\}}(\GD_1,\GD_1^*)$, where
$\Fb_1^*$ and $\Fb_2^*$ are two distinct dual OCRSFs of $\Fb$.
This means that
there exists $(F_1,F_1^*)\in\F^{\tau,(\Fb,\Fb_1^*)}(\GD_1,\GD_1^*)$ and
$(F_2,F_2^*)\in\F^{\tau,(\Fb,\Fb_2^*)}(\GD_1,\GD_1^*)$, such that $(F,F^*)$ is
obtained from each of $(F_i,F_i)^*$ by successively adding the same set
of oriented long edges $\Lb\setminus\Fb$ with essential moves. By Part $1.$
of Lemma \ref{lem:lem14}, if an essential move can be performed to add an
oriented long edge, then the reverse move can also be performed to recover the
original pair. Moreover, looking at Figure \ref{fig:fig19}, describing possible
essential reverse moves, we see that there is at most one essential reverse move
for removing a given oriented long edge. Thus, the same essential reverse moves
are performed on $(F,F^*)$ to recover $(F_1,F_1^*)$ and $(F_2,F_2^*)$, thus
implying that this is in fact the same pair. Since $\Fb_i^*$ is the restriction
to dual long edges of $F_i^*$ ($i=1,2$), we deduce that $\Fb_1^*=\Fb_2^*$.
\end{proof}

\subsection{The construction is weight preserving}
\label{sec:weightpreserving}

In this section, we state and prove Theorem \ref{prop:weightpres},
which establishes that the construction is weight preserving: we
show that the sum of the weights of all pairs of essential dual OCRSFs of
$\GD_1$ and $\GD_1^*$ in $\S^{(\Fb,\Fb^*)}(\GD_1,\GD_1^*)$ is equal, up to a
constant which only depends on the graph $G_1$, to the weight of the pair of
dual OCRSFs $(\Fb,\Fb^*)$ of $G_1$ and $G_1^*$. As a consequence we recover, by
an explicit computation, Theorem $2$ of \cite{BoutillierdeTiliere:iso_perio},
see Corollary \ref{cor:beaced}. Note that the constant could not be explicited
in the proof of Theorem $2$, but 
could only be recovered a posteriori in \cite{BoutillierdeTiliere:iso_gen}.

\begin{thm}\label{prop:weightpres}
Let $(\Fb,\Fb^*)$ be a pair of dual OCRSFs of $G_1$ and $G_1^*$. Then,
\begin{align*}
\sum_{(F,F^*)\in\S^{(\Fb,\Fb^*)}(\GD_1,\GD_1^*)}\left(
\prod_{(x,y)\in
F}f_x\overline{f_y}K_{x,y}\right)&(-\zs^{h_0(F)}\ws^{v_0(F)})^{\frac{1}{2}N(F,
F^*) } = \\
&=\mathbf{C}\left(\prod_{(\xb,\yb)\in
\Fb}\tan\theta_{\xb\yb}\right)(-\zs^{h_0(\Fb)}\ws^{v_0(\Fb)})^{\frac{1}{2}N(\Fb,
\Fb^*) },
\end{align*}
where $\mathbf{C}=2^{4|E(G_1)|+|V(G_1)|}\prod_{\xb\yb\in
E(G_1)}\sin^2\bigl(\frac{\theta_{\xb\yb}}{2}\bigr)\cos\theta_{\xb\yb}$.
\end{thm}
\begin{proof}
Using notations introduced in Section \ref{sec:conclusion}, let us recall
that:
$$
\S^{(\Fb,\Fb^*)}(\GD_1,\GD_1^*)=
\bigcup_{\tau\in\{cw,cclw\}^{V(G_1)}}\bigcup_{k=0}^m
\bigcup_{\{\{\eb_{i_1},\cdots,\eb_{i_k}\}:(i_1,\cdots,i_k)\in J_k\}}
\F^{\tau,(\Fb,\Fb^*),\{\eb_{i_1},\cdots,\eb_{i_k}\}}(\GD_1,\GD_1^*).
$$ 
By Proposition \ref{lem:lem14}, we know that for every pair $(F,F^*)$ of dual
essential
OCRSFs of $\S^{(\Fb,\Fb^*)}(\GD_1,\GD_1^*)$:
$$
(-\zs^{h_0(F)}\ws^{v_0(F)})^{\frac{1}{2}N(F,F^*)}=
(-\zs^{h_0(\Fb)}\ws^{v_0(\Fb)})^{\frac{1}{2}N(\Fb,\Fb^*)},
$$
so that we only need to handle the weights which do not involve the
coefficients $\zs$ and~$\ws$. Consider $\tau\in\{cw,cclw\}^{V(G_1)}$, a fixed
assignment of types to
vertices of $G_1$, and let us denote by $\Wc_k$ the
weighted sum (excluding $\zs$ and $\ws$) of configurations of
$\F^{\tau,(\Fb,\Fb^*),\{\eb_{i_1},\cdots,\eb_{i_k}\}}(\GD_1,\GD_1^*)$. We now
compute $\Wc_k$ by induction on $k$. 

In all computations below, we use the definition of the Kasteleyn orientation
of Section~\ref{sec:dimercharact}, the definition of the function $(f_x)_{x\in
V(\GD_1)}$ and of the rhombus half-angles in $\RR/4\pi\ZZ$ of Section
\ref{sec:421}.

\underline{\textbf{Computation of $\Wc_0$}}. Recall that
$\F^{\tau,(\Fb,\Fb^*)}(\GD_1,\GD_1^*)$ consists of all pairs $(F,F^*)$ of dual
OCRSFs of $\GD_1$ and $\GD_1^*$, such that $F$ is an essential OCRSF compatible
with $\Fb$ and $\tau$, characterized in Lemma \ref{lem:lem12}, and $F^*$ is the
unique dual OCRSF whose restriction to dual long edges is $\Fb^*$. Recall also
that, for every decoration $\xb$ of $\GD_1$, $v_1(\xb)$ denotes the unique root
vertex induced by $\Fb$.
\begin{enumerate}
 \item \underline{Contribution of long edges of $\Fb$}. Let $(\xb,\yb)$ be an
oriented edge of $\Fb$, and let $(v_k(\xb),v_\l(\yb))$ be the corresponding
long edge of $\GD_1$. Denote by $\theta_{\xb\yb}$ the
rhombus half-angle of $\xb\yb$. Then the contribution of this edge to $\Wc_0$
is:
$$
f_{v_k}\overline{f_{v_\l}}K_{v_k,v_\l}.
$$
By definition of the Kasteleyn matrix $K$, and of the vectors $f$, we have:
\begin{align*}
&K_{v_k,v_\l}=\eps_{v_k,v_\l}\cot\frac{\theta_{\xb\yb}}{2}\\
&f_{v_k}=e^{-i\frac{\alpha_{w_k}}{2}}-e^{-i\frac{\alpha_{z_k}}{2}},
\quad
f_{v_\l}=e^{-i\frac{\alpha_{w_\l}}{2}}-e^{-i\frac{\alpha_{z_\l}}{2}}.
\end{align*}
Moreover, by definition of the rhombus half-angles in $\RR/4\pi\ZZ$:
\begin{equation*}
f_{v_\l}=i\eps_{v_k,v_\l}f_{v_k},
\end{equation*}
so that $\overline{f_{v_\l}}=-i\eps_{v_k,v_\l}\overline{f_{v_k}}$. As
a consequence:
\begin{align}\label{eq:longedge}
f_{v_k}\overline{f_{v_\l}}K_{v_k,v_\l}&=
-i\cot\frac{\theta_{\xb\yb}}{2}\,|f_{v_k}|^2
=-i\cot\frac{\theta_{\xb\yb}}{2}\,2(1-\cos\theta_{\xb\yb})\nonumber\\
&=-2i\sin\theta_{\xb\yb}.
\end{align}

\item \underline{Contribution of inner edges of decorations}. Let $\xb$ be
a decoration of $\GD_1$, and let $w_j(\xb)z_{j+1}(\xb)$ be a a generic inner
edge of $\xb$. 

$\circ$ If $\tau(\xb)=cw$, then the edge is oriented from $z_{j+1}$ to
$w_j$, and its contribution to $\Wc_0$ is:
\begin{equation*}
\eps_{z_{j+1},w_j}f_{z_{j+1}}\overline{f_{w_j}}. 
\end{equation*}
Moreover, by definition of the angles in $\RR/4\pi\ZZ$, we have
$f_{z_{j+1}}=-\eps_{w_j,z_{j+1}}f_{w_j}$, so that the contribution is:
\begin{equation*}
-\eps_{z_{j+1},w_j}\eps_{w_j,z_{j+1}}f_{w_j}\overline{f_{w_j}}=1.
\end{equation*}
$\circ$ If $\tau(\xb)=cclw$, then the edge is oriented from $w_j$ to
$z_{j+1}$, and its contribution to $\Wc_0$ is:
\begin{equation*}
\eps_{w_j,z_{j+1}}f_{w_j}\overline{f_{z_{j+1}}}=-1. 
\end{equation*}

\item \underline{Contribution of triangles of root vertices}. Let $\xb$ be a
decoration of $\GD_1$, and $t_1(\xb)$ be the triangle of the root vertex
$v_1(\xb)$. Denote by $\theta_1(\xb)$ the rhombus half-angle of the long edge
incident to $v_1(\xb)$.

$\circ$ If $\tau(\xb)=cw$, then the triangle $t_1$ contains the edge
$(w_1,v_1)$, so that its contribution to $\Wc_0$ is:
\begin{equation*}
f_{w_1}\overline{f_{v_1}}K_{w_1,v_1}=e^{-i\frac{\alpha_{w_1}}{2}}(e^{i\frac{
\alpha_{w_1}}{2}}-e^{i\frac{\alpha_{z_1}}{2}})=1-e^{-i\theta_1}=
2i\sin\Bigl(\frac{\theta_1}{2}\Bigr)e^{-i\frac{\theta_1}{2}}.
\end{equation*}
$\circ$ If $\tau(\xb)=cclw$, then the triangle $t_1$ contains the edge
$(z_1,v_1)$, so that its contribution to $\Wc_0$ is:
\begin{equation*}
f_{z_1}\overline{f_{v_1}}K_{z_1,v_1}=e^{-i\frac{\alpha_{z_1}}{2}}(e^{i\frac{
\alpha_{w_1}}{2}}-e^{i\frac{\alpha_{z_1}}{2}})=-1+e^{i\theta_1}
=2i\sin\Bigl(\frac{\theta_1}{2}\Bigr)e^{i\frac{\theta_1}{2}}.
\end{equation*}

\item \underline{Contribution of triangles of non-root vertices}. Let $\xb$ be
a decoration of $\GD_1$, $t_j(\xb)$ ($j\neq 1$) be a triangle of a non-root
vertex, and $\theta_j$ be the rhombus half-angle of the long edge incident to
$v_j(\xb)$. Then, $t_j(\xb)$ contains any of the three possible $2$-edge
configurations, with the orientation induced by the type. Thus,

$\circ$ If $\tau(\xb)=cw$, its contribution to $\Wc_0$ is, see also Figure
\ref{fig:fig19}:
\begin{align}\label{equ:3}
&f_{w_j}\overline{f_{v_j}}K_{w_j,v_j}f_{v_j}\overline{f_{z_j}}K_{v_j,z_j}+
f_{v_{j}}\overline{f_{z_j}}K_{v_j,z_j}f_{w_j}\overline{f_{z_j}}K_{w_j,z_j}+
f_{v_j}\overline{f_{w_j}}K_{v_j,w_j}f_{w_j}\overline{f_{z_j}}K_{w_j,z_j}
=\nonumber\\
&=f_{w_j}\overline{f_{z_j}}f_{v_j}(\overline{f_{v_j}}-\overline{f_{z_j}}
+\overline{f_{w_j}})\nonumber\\
&=f_{w_j}\overline{f_{z_j}}f_{v_j}(2\overline{f_{w_j}})\quad\text{ (by
definition of the vector $f$)}\nonumber\\
&=2\overline{f_{z_j}}f_{v_j}=-2e^{i\frac{\alpha_{z_j}}{2}}(e^{-i\frac{\alpha_{
w_j}}{2}}-e^{-i\frac{\alpha_{z_j}}{2}})
=2(1-e^{-i\theta_j})\nonumber\\
&=4i\sin\Bigl(\frac{\theta_j}{2}\Bigr)e^{-i\frac{\theta_j}{2}}.
\end{align}
$\circ$ If $\tau(\xb)=cclw$, its contribution to $\Wc_0$ is:
\begin{align}\label{equ:4}
(-1)f_{v_j}\overline{f_{w_j}}f_{z_j}\overline{f_{w_j}}+
f_{v_{j}}\overline{f_{z_j}}&f_{z_j}\overline{f_{w_j}}+(-1)f_{z_j}\overline{f_{
v_j }}(-1)f_{v_j}\overline{f_{w_j}}=\nonumber\\
&=\overline{f_{w_j}}f_{z_j}f_{v_j}(-\overline{f_{w_j}}+\overline{f_{z_j}}
+\overline{f_{v_j}})\nonumber\\
&=\overline{f_{w_j}}f_{z_j}f_{v_j}(2\overline{f_{z_j}})\quad\text{ (by
definition of the vectors $f$)}\nonumber\\
&=2\overline{f_{w_j}}f_{v_j}=2e^{i\frac{\alpha_{w_j}}{2}}(e^{-i\frac{\alpha_{w_j
}}{2}}-e^{-i\frac{\alpha_{z_j}}{2}})
=2(1-e^{i\theta_j})\nonumber\\
&=-4i\sin\Bigl(\frac{\theta_j}{2}\Bigr)e^{i\frac{\theta_j}{2}}.
\end{align}
\end{enumerate}
Combining this, we deduce the contribution of a decoration $\xb$. Note that in
computations below, we use the fact that
$\sum_{j=1}^{d_{\xb}}\frac{\theta_j}{2}=\frac{\pi}{2}$.

$\circ$ If $\tau(\xb)=cw$, the contribution of a decoration $\xb$ to $\Wc_0$ is:
\begin{equation*}
\left(\prod_{j=1}^{d_{\xb}}2i\sin\Bigl(\frac{\theta_j}{2}\Bigr)e^{-i\frac{
\theta_j}{2}} \right)2^{d_{\xb}-1}
=-2^{2d_{\xb}-1}i^{d_{\xb}+1}\prod_{j=1}^{d_{\xb}}\sin\Bigl(\frac{\theta_j}{2}
\Bigr).
\end{equation*}
$\circ$ If $\tau(\xb)=cclw$, the contribution of a decoration $\xb$ to $\Wc_0$
is:
\begin{equation*}
\left(\prod_{j=1}^{d_{\xb}}(-1)2i\sin\Bigl(\frac{\theta_j}{2}\Bigr)e^{i\frac{
\theta_j}{2}} \right)(-2)^{d_{\xb}-1}
=-2^{2d_{\xb}-1}i^{d_{\xb}+1}\prod_{j=1}^{d_{\xb}}\sin\Bigl(\frac{\theta_j}{2}
\Bigr).
\end{equation*}
We deduce that the contribution of a decoration is in fact independent of its
type. Let us denote by 
$$
N=V(G_1),\;M=E(G_1).
$$
Then, since $\Fb$ is an OCRSF of $G_1$, it has $N$ edges. Taking the product
over all long edges of $\Fb$ yields a contribution:
\begin{equation}\label{equ:1}
(-2)^N i^N \prod_{(\xb,\yb)\in\Fb}\sin\theta_{\xb\yb},
\end{equation}
where the contribution is in fact independent of the orientation of the edges.
Observing that $\sum_{\xb\in V(G_1)}d_{\xb}=2M$, and taking the product over
all decorations of $\GD_1$ yields a contribution:
\begin{equation}\label{equ:2}
(-1)^N 2^{4M-N}i^{2M+N}\prod_{\xb\yb\in
E(G_1)}\sin^2\Bigl(\frac{\theta_{\xb\yb}}{2}\Bigr). 
\end{equation}
Taking the product of Equations \eqref{equ:1} and \eqref{equ:2}, gives:
\begin{equation}\label{equ:5}
\Wc_0=\left[2^{4M}(-1)^{M+N}\prod_{\xb\yb\in
E(G_1)}\sin^2\Bigl(\frac{\theta_{\xb\yb}}{2}\Bigr)\right]
\prod_{\xb\yb\in\Fb}\sin\theta_{\xb\yb}.
\end{equation}

\underline{\textbf{Computation of $\Wc_k,\,k\geq 1$}}. Recall that 
$\F^{\tau,(\Fb,\Fb^*),\{\eb_{i_1},\cdots,\eb_{i_k}\}}(\GD_1,\GD_1^*)$ consists
of all pairs $(F,F^*)$ of dual essential OCRSFs of $\GD_1$ and $\GD_1^*$,
obtained by performing essential moves adding the edge $+\eb_{i_k}$ or
$-\eb_{i_k}$, on pairs of dual OCRSFs of
$\F^{\tau,(\Fb,\Fb^*),\{\eb_{i_1},\cdots,\eb_{i_{k-1}}\}}(\GD_1,\GD_1^*)$. Using
Lemma \ref{lem:lem12}, characterizing the set
$\F^{\tau,(\Fb,\Fb^*)\}}(\GD_1,\GD_1^*)$, and returning to the definition of
essential moves, see Figure \ref{fig:fig18}, we deduce by induction that, for
every $k\in\{0,\cdots,m\}$ and for every $(i_1,\cdots,i_k)\in J_k$, the
restriction to a decoration $\xb$ of $\GD_1$ of the first component $F$ of a
pair $(F,F^*)$ of dual OCRSFs of
$\F^{\tau,(\Fb,\Fb^*),\{\eb_{i_1},\cdots,\eb_{i_k}\}}(\GD_1,\GD_1^*)$ contains:
\begin{enumerate}
\item all inner edges with the orientation induced by the type $\tau(\xb)$,
\item any of the three possible $2$-edge configurations at triangles of
non-root vertices,
\item the edge $(w_1,v_1)$ (resp.
$(z_1,v_1)$) at the triangle $t_1(\xb)$, if $\tau(\xb)=cw$
(resp. $\tau(\xb)=cclw$),
\item one or two of the $1$-edge configurations at triangles of other root
vertices, depending on whether one or two essential moves can be performed.
\end{enumerate}
Let us compute the ratio $\frac{\Wc_k}{\Wc_{k-1}}$. To simplify notations, we
denote by $\xb\yb$ the edge $\eb_{i_k}$, by $v_k(\xb)v_\l(\yb)$ the
corresponding long edge of $\GD_1$, and by $\theta$ the corresponding
rhombus half-angle.

The edge $\eb_{i_k}$ is absent in OCRSFs of
$\F^{\tau,(\Fb,\Fb^*),\{\eb_{i_1},\cdots,\eb_{i_{k-1}}\}}(\GD_1,\GD_1^*)$, so
that $v_k$ and $v_\l$ are non-root vertices. By Point $3.$, this implies that
the
triangles $t_k$ and $t_{\l}$ each contain any of the three possible
$2$-edge configurations. Using equations \eqref{equ:3} and \eqref{equ:4}, we
deduce that the contribution to $\Wc_{k-1}$ of the absent edge and of the two
incident triangles is:
\begin{itemize}
\item If $\tau(\xb)=cw$ and $\tau(\yb)=cw$,
\begin{equation*}
(4i\sin\Bigl(\frac{\theta}{2}\Bigr)e^{-i\frac{\theta}{2}})^2=
-16\sin^2\Bigl(\frac{\theta}{2}\Bigr)e^{-i\theta}.
\end{equation*}
\item If $\tau(\xb)=cw$ and $\tau(\yb)=cclw$, or $\tau(\xb)=cclw$ and
$\tau(\yb)=cw$,
\begin{equation*}
16\sin^2\Bigl(\frac{\theta}{2}\Bigr).
\end{equation*}
\item If $\tau(\xb)=cclw$ and $\tau(\yb)=cclw$,
\begin{equation*}
-16\sin^2\Bigl(\frac{\theta}{2}\Bigr)e^{i\theta}.
\end{equation*}
\end{itemize}

We now describe what happens when performing an essential move adding the long
edge $(v_k,v_\l)$ or $(v_\l,v_k)$. Suppose that the dual edge in $F^*$ is
oriented as in Case A of Figure \ref{fig:fig18}. We give all details for the
first case only, since others are similar.
\begin{itemize}
 \item If $\tau(\xb)=cw$ and $\tau(\yb)=cw$. When the essential move adds the
edge $v_k v_\l$ in either of the two directions, one of $v_k$ or $v_\l$ is a
non-root vertex, so that the contribution of $t_k$ or $t_\l$ is:
$$
4i\sin\Bigl(\frac{\theta}{2}\Bigr)e^{-i\frac{\theta}{2}}.
$$
By Equation \eqref{eq:longedge}, the contribution of the edge $(v_k,v_\l)$
is independent
of its orientation and is equal to:
$$
-2i\sin\theta.
$$
Suppose that the essential move adds the long edge $(v_k,v_\l)$. Then, by Case
$A$ of Figure \ref{fig:fig18}, there are two possible configurations at the
triangle $t_k$, and the contribution is:
$$
f_{w_k}\overline{f_{v_k}}K_{w_k,v_k}+f_{w_k}\overline{f_{z_k}}K_{w_k,z_k}=
f_{w_k}\overline{f_{v_k}}-f_{w_k}\overline{f_{z_k}}=f_{w_k}\overline{f_{w_k}}=1.
$$
Suppose that the essential move adds the long edge $(v_\l,v_k)$. Then, using
symmetries, we are in Case $B$ of Figure \ref{fig:fig18}. Thus there is one
possible configuration at the triangle $t_\l$, and the contribution is:
$$
f_{w_\l}\overline{f_{z_\l}}K_{w_\l,z_\l}=e^{-i\theta}.
$$
As a consequence, the contribution to $\Wc_k$ of the edge added in one of the
two possible directions and of the two incidents triangles is
\begin{equation*}
8\sin\Bigl(\frac{\theta}{2}\Bigr)e^{-i\frac{\theta}{2}}\sin\theta(1+e^{-i\theta}
)
=16\sin\Bigl(\frac{\theta}{2}\Bigr)e^{-i\theta}\sin\theta\cos\Bigl(\frac{\theta}
{2}\Bigr) .
\end{equation*}
\item If $\tau(\xb)=cw$ and $\tau(\yb)=cclw$. In a similar way, the
contribution to $\Wc_k$ is:
\begin{align*}
&\left(-4i\sin\Bigl(\frac{\theta}{2}\Bigr)e^{i\frac{\theta}{2}}
\right)(-2i\sin\theta)(1)+
\left(4i\sin\Bigl(\frac{\theta}{2}\Bigr)e^{-i\frac{\theta}{2}}
\right)(-2i\sin\theta)(-1)\\
&=-16\sin\Bigl(\frac{\theta}{2}\Bigr)\sin\theta\cos\Bigl(\frac{\theta}{2}\Bigr).
\end{align*}
\item If $\tau(\xb)=cclw$ and $\tau(\yb)=cw$, the
contribution to $\Wc_k$ is:
\begin{align*}
&\left(4i\sin\Bigl(\frac{\theta}{2}\Bigr)e^{-i\frac{\theta}{2}}
\right)(-2i\sin\theta)(-e^{i\theta})+
\left(-4i\sin\Bigl(\frac{\theta}{2}\Bigr)e^{i\frac{\theta}{2}}
\right)(-2i\sin\theta)(e^{-i\theta})\\
&=-16\sin\Bigl(\frac{\theta}{2}\Bigr)\sin\theta\cos\Bigl(\frac{\theta}{2}\Bigr).
\end{align*}
\item If $\tau(\xb)=cclw$ and $\tau(\yb)=cclw$, the
contribution to $\Wc_k$ is:
\begin{align*}
&\left(-4i\sin\Bigl(\frac{\theta}{2}\Bigr)e^{i\frac{\theta}{2}}
\right)(-2i\sin\theta)(-e^{i\theta})+
\left(-4i\sin\Bigl(\frac{\theta}{2}\Bigr)e^{i\frac{\theta}{2}}
\right)(-2i\sin\theta)(-1)\\
&=16\sin\Bigl(\frac{\theta}{2}\Bigr)e^{i\theta}\sin\theta\cos\Bigl(\frac{\theta}
{ 2 } \Bigr).
\end{align*}
\end{itemize}
Moreover, by Points $1.-4.$ we know that the contribution of all other long
edges, inner edges and triangles is the same in $\Wc_k$ and $\Wc_{k-1}$. Thus,
from the above computations we deduce that, independently of the type of the
decorations $\xb$ and $\yb$:
\begin{align}\label{equ:final}
\frac{\Wc_k}{\Wc_{k-1}}&=-\frac{16\sin\bigl(\frac{\theta}{2}\bigr)
\sin\theta\cos\bigl(\frac{\theta}
{ 2 } \bigr)}{16\sin^2\bigl(\frac{\theta}{2}\bigr)}\nonumber\\
&=-\sin\theta\cot\bigl(\frac{\theta}
{ 2 } \bigr)=-(1+\cos\theta).
\end{align}
Let us now introduce the notation $\theta_{i_j}$ for the rhombus half-angle of
the edge $\eb_{i_j}$. Then, from Formula \eqref{equ:final}, we deduce that
independently of the type of the vertices:
$$
\Wc_k=\Wc_0\prod_{j=1}^k(-(1+\cos\theta_{i_j})). 
$$
\underline{\textbf{Conclusion}}. Independently of the type of the vertices, we
have:
\begin{align*}
\sum_{k=0}^m \sum_{(i_1,\cdots,i_k)\in J_k}\Wc_k
&=\Wc_0\sum_{k=0}^m \sum_{(i_1,\cdots,i_k)\in J_k}\prod_{j=1}^k
(-(1+\cos\theta_{i_j}))\\
&=\Wc_0
\prod_{i=1}^m[1-(1+\cos\theta_i)]\\
&=\Wc_0(-1)^m\prod_{i=1}^m\cos\theta_i\\
&=\Wc_0(-1)^{M-N} \prod_{\xb\yb\in
E(G_1)\setminus \Fb}\cos\theta_{\xb\yb},\;\text{ (changing notations)}\\
&=\left[2^{4M}\prod_{\xb\yb\in
E(G_1)}\sin^2\Bigl(\frac{\theta_{\xb\yb}}{2}\Bigr)\right]
\Bigl[\prod_{\xb\yb\in\Fb}\sin\theta_{\xb\yb}\Bigr]
\Bigl[\prod_{\xb\yb\in E(G_1)\setminus \Fb}\cos\theta_{\xb\yb}\Bigr]
,\;\text{ (by
\eqref{equ:5})}\\
&=\left[2^{4M}\prod_{\xb\yb\in
E(G_1)}\sin^2\Bigl(\frac{\theta_{\xb\yb}}{2}\Bigr)\cos\theta_{\xb\yb}\right]
\prod_{\xb\yb\in\Fb}\tan\theta_{\xb\yb}
\end{align*}
Summing over the $2^{N}$ possible types for the vertices of $G_1$ yields
Theorem \ref{prop:weightpres}.
\end{proof}

\begin{cor}[\cite{BoutillierdeTiliere:iso_perio,BoutillierdeTiliere:iso_gen}]
\label{cor:beaced}
$$
\Pdimer(\zs,\ws)=\Big(2^{|V(G_1)|}\prod_{\xb\yb\in E(G_1)}[
\cot^2\Bigl(\frac{\theta_{\xb\yb}}{2}\Bigr)-1]\Bigr)\Plap(\zs,\ws).
$$
\end{cor}

\begin{proof}
From Theorem \ref{thm:correspondence} and \ref{prop:weightpres}, we deduce that:
\begin{equation*}
\Pdimer^0(\zs,\ws)=\mathbf{C}\Plap(\zs,\ws). 
\end{equation*}
Moreover, by Equation \eqref{equ:30}, 
$$
\Pdimer^0(\zs,\ws)=\Bigl(\prod_{x\in
V(\GD_1)}|f_x|^2\Bigr)\Pdimer(\zs,\ws). 
$$
By definition of the vector $f$, for every decoration $\xb$ and for every
$k\in\{1,\cdots,d_{\xb}\}$:
\begin{align*}
|f_{w_k(\xb)}|&=|f_{z_k(\xb)}|=1\\
|f_{v_k(\xb)}|^2&=2(1-\cos\theta)=4\sin^2\Bigl(\frac{\theta}{2}\Bigr),
\end{align*}
where $\theta$ is the rhombus half-angle of the long edge incident to
$v_k(\xb)$. As a consequence:
$$\prod_{x\in V(\GD_1)}|f_x|^2=2^{4M}\prod_{\xb\yb\in
E(G_1)}\sin^4\Bigl(\frac{\theta_{\xb\yb}}{2}\Bigr).$$
We conclude:
\begin{align*}
\Pdimer(\zs,\ws)&=\Bigl(\frac{\mathbf{C}}{\prod_{x\in
V(\GD_1)}|f_x|^2}\Bigr)\Plap(\zs,\ws)\\
&=\left(2^N\prod_{\xb\yb\in
E(G_1)}\frac{\cos\theta_{\xb\yb}}{\sin^2\Bigl(\frac{\theta_{\xb\yb}}{2}
\Bigr)}\right)\Plap(\zs,\ws)\\
&=\Big(2^N\prod_{\xb\yb\in E(G_1)}[
\cot^2\Bigl(\frac{\theta_{\xb\yb}}{2}\Bigr)-1]\Bigr)\Plap(\zs,\ws).
\end{align*}

\end{proof}

\bibliographystyle{alpha}
\bibliography{survey}

\end{document}